%% file: main.tex
\documentclass[final]{IEEEtran}
\usepackage{amsthm,amssymb,graphicx,multirow,amsmath,color,amsfonts,physics}%,ulem}
\usepackage[update,prepend]{epstopdf}
\usepackage[noadjust]{cite}
\usepackage{tikz}
\usepackage{bbm} % for \nbb1 
\usepackage{pdfpages}
\usepackage{balance}
\usepackage{multirow}
\usepackage{comment}
\usepackage{subfigure}
\allowdisplaybreaks % Allows breaking of eqnarray over multiple pages (avoids unnecessary blanks in the document before eqnarray)

\begin{document}
\include{notation}
\pagenumbering{gobble}
\graphicspath{{./Figures/}}
\title{Performance Evaluation of HAPS-enabled\\ Coverage Enhancement in Hard-to-Reach Areas}
\author{
 Hao Lin,~\IEEEmembership{Graduate Student Member,~IEEE}, Mustafa A. Kishk,~\IEEEmembership{Member,~IEEE}\\ and Mohamed-Slim Alouini,~\IEEEmembership{Fellow,~IEEE}
\thanks{Hao Lin is with the Electrical and Computer Engineering Program, Computer, Electrical and Mathematical Sciences and Engineering Division (CEMSE), King Abdullah University of Science and Technology (KAUST),
Thuwal 23955-6900, Saudi Arabia (e-mail: hao.lin.std@gmail.com).\\
\indent Mustafa A. Kishk is with the Department of Electronic Engineering Program,
Maynooth University, Maynooth, W23 F2H6 Ireland (e-mail:
mustafa.kishk@mu.ie).\\%National University of Ireland
\indent Mohamed-Slim Alouini is with the CEMSE Division, King Abdullah
University of Science and Technology (KAUST), Thuwal 23955-6900,
Saudi Arabia (e-mail: slim.alouini@kaust.edu.sa).}
}

\maketitle
\vspace{-1.8cm}%\vspace{-1.8cm}
\begin{abstract}
High altitude platform stations (HAPSs) are becoming a key component of future non-terrestrial networks (NTNs). HAPSs can serve a larger area than uncrewed aerial vehicles (UAVs) and offer lower propagation latency, maintenance expense, and energy costs than satellites. A major application of HAPSs is to serve the areas where terrestrial network (TN) deployment is infeasible, especially in hard-to-reach areas and post-disaster areas. For instance, in the Amazon rainforest, the Mediterranean region, and deserts, TN deployment is severely constrained by geographical and environmental conditions. Only areas close to transportation networks or coastlines can be covered, while large areas remain uncovered. Such coverage holes in hard-to-reach areas are typically overlooked in existing literature. Motivated by these
realistic cases, in this paper, we use tools from stochastic geometry to mathematically model hard-to-reach areas where cellular terrestrial infrastructure only exists at their perimeter. We propose to deploy a HAPS constellation over this hard-to-reach area to enhance connectivity. For that setup, we derive the downlink (DL) and uplink (UL) coverage performance of
the considered user equipment (UE) as a function of the location of the UE inside the coverage hole. Our results show how the number of HAPSs, beamwidth, and HAPS altitude affect the
DL and UL coverage probabilities. Finally, we provide multiple useful
guidelines for future HAPS deployment.
\end{abstract}
% \vspace{-0.7cm}
\begin{IEEEkeywords}
% \vspace{-0.3cm}
High altitude platform stations, stochastic geometry, non-terrestrial networks, hard-to-reach areas, coverage enhancement, downlink and uplink analysis
\end{IEEEkeywords}

\section{Introduction} \label{sec:Intro}
The International Telecommunication Union (ITU) defines high altitude platform stations (HAPSs) as radio stations located at altitudes of 20-50 kilometers \cite{HAPSITU}. Operating in the stratosphere, they can offer lower propagation latency and maintenance costs than satellite networks and can provide communication services over a wider geographical area than uncrewed aerial vehicles (UAVs) \cite{9380673,10355104}. To provide broadband access services to ubiquitous unmodified user devices, HAPSs can share the same spectrum with terrestrial networks (TNs) \cite{10082988}. Also, HAPSs can be launched rapidly to hotspots or post-disaster areas to provide stable support for emergency communications, and operate for many weeks or months using buoyancy and aerodynamics \cite{9900369}. The high payload capacity of HAPSs can support large antenna arrays, so they can realize directional beamforming for both ground and air users \cite{9380673,10930564}. In particular, massive multiple-input multiple-output (MIMO) and stacked intelligent metasurfaces (SIMs) can serve as promising physical-layer enablers for advanced beamforming in HAPS communications \cite{11219215,11456560}. By providing highly directional and adaptive beams, these techniques can enhance full-space connectivity for the mobile users and aerial platforms \cite{11474787}.\\
\indent HAPSs are integrated into International Mobile Telecommunications (IMT) networks, known as `HAPS as IMT base stations' (HIBS) \cite{9900369,9681623}. They can be used to bridge the digital divide and fill coverage gaps in TNs \cite{10474118,10504115}. In many places in the world, such as the Amazon rainforest, Mediterranean islands, deserts, and remote mountain regions, TNs are severely limited by geographical conditions and terrain. Terrestrial base station (TBS) deployment is close to transportation lines or coastlines, while vast areas still lack reliable communication coverage. Although these areas do not have a large residential population, they still contain many potential user equipment (UE) clusters, such as the photovoltaic energy industry, agricultural Internet of Things (IoT), desert greening projects, rainforest fire early warning systems, offshore fishing networks, and oil extraction networks. Therefore, HAPSs can provide stable communication support for these important industries. For areas where terrestrial communication systems are destroyed by natural disasters such as earthquakes, floods, or wildfires, HAPS networks can serve as an effective temporary supplement for quick recovery \cite{11387886}.\\
\indent It is important to evaluate the performance of HAPS networks in filling the coverage hole in the hard-to-reach areas, considering the coexistence with spatially constrained TNs and directional beamforming ability of HAPSs. Stochastic geometry has been widely used in the performance analysis of wireless systems, without losing tractability and accuracy compared to real-world networks
\cite{haenggi2013stochastic}. Therefore, to bridge the coverage gaps in the hard-to-reach areas, we construct a stochastic geometry-based framework considering HAPSs with directional antennas and provide the performance analysis in this paper.

\subsection{Related Work}
In this paper, we focus on bridging coverage holes in hard-to-reach areas with limited terrestrial networks, and aim to analyze the uplink (UL) and downlink (DL) transmission performance of HAPS networks considering directional beamforming. Therefore, we divide the relevant work into: \textit{i) \textbf{HAPSs for Wireless Networks}} and \textit{ii) \textbf{Stochastic Geometry for HAPS Networks}}. 

\textbf{\textit{HAPSs for Wireless Networks}}: In the 6G era, the integration of HAPSs introduces an additional degree of freedom to reduce energy consumption. For example, \cite{10571153,10304250} proved that HAPSs can work as a promising solution to realize energy efficient wireless networks. Authors in \cite{10608095} proposed a cylindrical antenna for HAPS that utilizes vertical linear array sectors. They also proposed a non-orthogonal multiple access clustering method to address the high spatial correlation among users and developed an algorithm to maximize spectral efficiency and energy efficiency while meeting quality of service. Authors in \cite{10949721} formulated an optimization problem to maximize the total capacity of ground UEs, and showed that lower altitudes are preferred when the channel conditions deteriorate. In \cite{11363370}, authors proposed a reconfigurable intelligent surface (RIS)-assisted rate splitting multiple access system of HAPSs, to maximize the sum rate by jointly optimizing precoding, RIS beamforming and time allocation under secrecy rate constraints. In addition, stacked intelligent metasurfaces can be applied to enable advanced wave-domain transceiver design \cite{10900449}. For instance, authors in \cite{11456560} derived closed-form expressions for the outage probability in SIM-assisted HAPS networks and proposed an alternating optimization framework and an unsupervised deep neural network for energy efficiency optimization. In \cite{9773096}, authors showed that integrating HAPSs with TNs can help offload connection tasks, reduce energy consumption of TBSs and make the whole network more sustainable. Therefore, HAPSs can complement terrestrial networks in supporting massive dynamic and unpredictable traffic demands in urban areas \cite{10192297}. Authors in \cite{gautam2026reliability} focused on the 5G New Radio (NR) communication system with HAPSs, low Earth orbit (LEO) satellites and TNs, and evaluated the reliability function and mean time to failure. They studied the cases with repairable or non-repairable k-out-of-n models for the HAPS networks. The integration of space networks, air networks and ground networks has been proved to be effective. Authors in \cite{9869801} took the maritime networks into account and analyzed the challenges and difficulties in realizing global coverage. The interference problem is the bottleneck for the operation of integrated networks. Authors in \cite{9904854} designed a joint power-subcarrier allocation scheme, and developed a rapidly converging iterative algorithm to solve the max-min fairness optimization problem of all UEs.

By integrating radio frequency and optical communication components, HAPSs can form an integrated network across air, space, and ground, providing high-speed services to remote areas. Authors in \cite{10599115} presented the complex networking challenges faced by the space-air-ground integrated network (SAGIN) and learning-based solutions. In \cite{mahmoud2025synthesize}, authors discussed the integration of HAPSs with deep learning models. They integrated the deep neural network (DNN) with a modified gravitational search algorithm and particle swarm optimization algorithm to effectively cover diverse road paths in traversing challenging terrains like deserts, mountain chains, and forested areas. By integrating deep reinforcement learning (DRL) and convex optimization, authors in \cite{11083742} jointly optimized dynamic caching strategies and resource allocation across the whole HAPS network, and provided a practical solution to enhance rural connectivity and bridge the digital divide. Through a multi-agent deep Q learning method, authors in \cite{s22041630} addressed the impact of HAPS transmission power on terrestrial networks when sharing the same spectrum. They minimized the outage probability of downlink transmission when satisfying the interference constraint. In \cite{10016705}, authors discussed in detail the applications of reinforcement learning (RL) in cellular-connected non-terrestrial networks (NTNs) and NTN-aided wireless communications and surveyed the literature for different RL formulations to solve NTN control problems.

In addition to working as super macro base stations (SMBSs), HAPSs can also act as multimode aerial nodes, including sensing stations, computation nodes, relays, and RIS platforms in the sky \cite{10186454}. For instance, authors in \cite{11429523} introduced an application of integrated sensing and communication systems. They jointly optimized the HAPS deployment strategy together with transmit beamforming, to simultaneously deliver communication services to multiple users and conduct synthetic aperture radar imaging for ground targets. HAPSs can also provide backhaul service for low altitude platforms (LAPs) in post-disaster areas. Authors in \cite{10097717} showed that, when the malfunction area exceeds a certain bound, it is necessary to deploy a HAPS to assist the backhaul of LAPs and provide better services for ground UEs. In \cite{11006480}, authors proposed a hierarchical aerial multi-access edge computing (MEC) system, where a HAPS provides stable and strong computing services for multi-UAVs to connect the remote ground users. Authors in \cite{9714482} constructed a hierarchical aerial computing framework of HAPSs and UAVs to provide MEC services for various IoT devices, and proposed effective algorithms for offloading decisions. In vehicular networks, HAPSs together with intelligent connected vehicles and roadside units can provide smooth distributed computing and avoid interruptions caused by handoffs \cite{10103832}. HAPSs can help reduce the delay of the system by optimizing the computation offloading, caching decisions and resource allocations \cite{9772280}. In \cite{10254352,abderrahim2025green}, authors proposed a flying data center by HAPS to support the operation of terrestrial data centers. They demonstrated that a data center-enabled HAPS can meet the growing energy demands and reduce carbon footprint due to the naturally low temperature in the stratosphere and the ability to harvest solar energy.

\textbf{\textit{Stochastic Geometry for HAPS Networks}}: Stochastic geometry has been widely used to evaluate the performance of wireless networks, including TNs, NTNs and integrated systems \cite{shi2024stochastic}. In \cite{10634042}, authors studied the system-level metrics of NTNs, where HAPSs play an important role in future networks to obtain better line-of-sight (LoS) conditions and higher availability. They discussed the coverage-based, relay-based and routing-based metrics and showed how HAPS deployment affects the coverage performance and energy efficiency of the entire NTNs. In \cite{10082988}, authors showed that low height and suitable deployment density can improve the coverage probability and transmission capacity of a HAPS network. Authors in \cite{11363148} emphasized the importance to apply a HAPS constellation in islands and maritime areas to enhance the connectivity. They proposed using the terrain information to divide the HAPSs into two regions with different LoS conditions and build a unified framework to evaluate the downlink coverage probability for onshore, nearshore and offshore UEs. The authors in \cite{11389817} proposed deploying an aerodynamic HAPS over the urban area, operating along a fixed circular trajectory, and the UEs calling to the HAPS are modeled as a hard-core binomial point process (BPP). They considered the antenna radiation patterns and the space division multiple access (SDMA) scheme, and calculated the long-term downlink coverage probability and ergodic capacity. In \cite{10930564}, authors modeled the global HAPS network as a spherical Poisson point process (PPP) and derived the expressions for uplink and downlink coverage probability considering the directional beamforming.

Stochastic geometry is also important to evaluate the feasibility of HAPSs in future TN-NTN integrated systems. In \cite{10506977}, authors highlighted that the HAPSs and LEO satellites are favorable when the cost is constrained and underlined their efficacy in balancing cost and performance in network design. Authors in \cite{11251052} conducted the downlink coverage analysis for a heterogeneous NTN including uncrewed aerial vehicles (UAVs), HAPSs and LEO satellites. They also evaluated the relationship between the expected delay and the densities of different NTN nodes. In \cite{10341311}, authors proposed a multi-layer heterogeneous SAGIN network considering multi-band Terahertz/Radio Frequency (THz/RF) channel models, derived the joint signal-to-interference-plus-noise ratio (SINR) and rate coverage probability and found the optimal deployment scheme for SAGIN. Authors in \cite{10040542} considered the maritime applications and evaluated the coverage performance of surface stations through the cooperation between onshore stations, tethered balloons, HAPSs and satellites. HAPSs can also work as decode-and-forward (DF) relays between LEO satellites and ground users. Authors in \cite{11106495} modeled the HAPSs as a PPP and analyzed the reliability performance of the satellite-air-terrestrial system. They derived the approximate closed-form expression for the outage probability based on orthogonal time frequency space (OTFS) technique. 

In \cite{11363148}, we investigated the feasibility of using HAPSs to cover open ocean areas without terrestrial infrastructure, and discussed the quality-of-service fairness among ubiquitous users under different channel environments. In \cite{11052265}, we discussed how HAPSs and gateways can form a large-scale continuous service area for IoT devices and mobile users. However, unlike vast maritime areas such as the Pacific Ocean, some inland or nearshore hard-to-reach areas are uncovered due to geographic and environmental constraints, where terrestrial base stations can only be deployed around their boundaries. In this manuscript, we further investigate the feasibility of HAPSs with the ITU radiation pattern to cover these blind areas and characterize the association and mutual interference between HAPSs and spatially constrained TBSs in realistic hard-to-reach areas.

\subsection{Contributions}
The contributions of this paper can be summarized as follows:
\begin{itemize}
    \item We build a mathematical framework for the HAPS-based solutions in hard-to-reach areas, addressing the impact of directional beamforming of HAPSs and coexistence of a HAPS network and terrain-constrained TNs.
    \item We investigate the downlink and uplink performance of HAPS networks in hard-to-reach areas and verify the feasibility of using HAPSs to bridge the digital divide. We also quantitatively analyze the effect of the number of HAPSs, 3 dB beamwidth and HAPS altitude on the DL and UL coverage performance.
    \item We provide design guidelines for the HAPS-based solution in hard-to-reach areas. Our results reveal how the 3 dB beamwidth of directional antennas and the number of HAPSs should be jointly designed considering both DL and UL performance.
\end{itemize}

\section{System Model} \label{sec:SysMod}

\begin{figure*}
    \centering
    \includegraphics[width=0.6\linewidth]{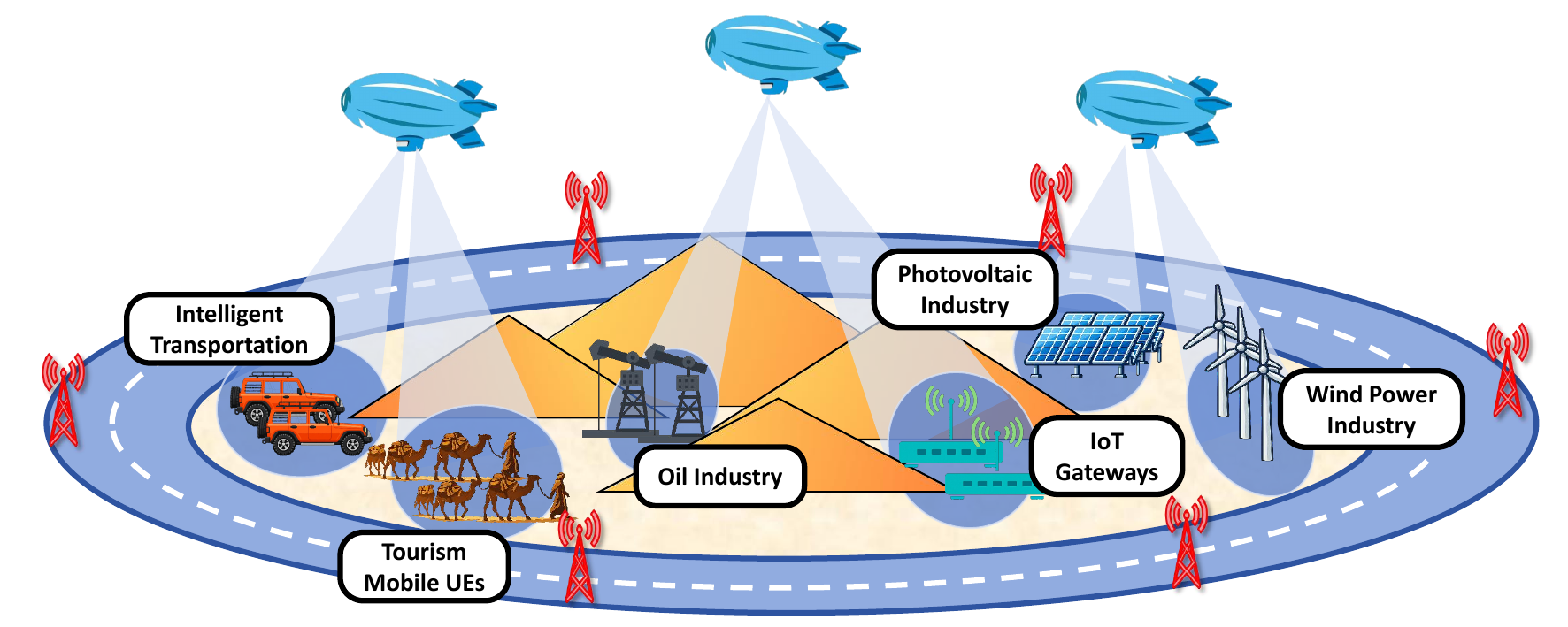}
    \caption{Illustration of the HAPS-based solution for a hard-to-reach area, taking desert applications as an example.}
    \label{fig:SystemModel}
\end{figure*}

\begin{table*}[t]\caption{Table of Notations}
\centering
\begin{center}
\resizebox{\textwidth}{!}{
\renewcommand{\arraystretch}{1}%1.4
    \begin{tabular}{ {c} | {l} }
    \hline
        \hline
    \textbf{Notation} & \textbf{Description} \\ \hline
    $\textbf{o}_T$; $\textbf{o}_H$; $\textbf{o}_U$ & The center of TBS deployment area; the center of HAPS deployment area and the center of the hard-to-reach area.\\ \hline
    $r_{\rm in}$; $r_{\rm out}$ & The radius of the hard-to-reach area; the maximum distance from the TBSs to the center of the hard-to-reach area.\\ \hline
    $\Phi_T$; $\Phi_H$; $\Phi_U$ & The sets of locations of TBSs, HAPSs, and grant-free interfering UEs active on the considered uplink resource block, respectively.\\ \hline
    $\lambda_T$; $N_H$; $\lambda_U$ & The density of TBSs; the number of HAPSs; the density of grant-free interfering UEs active on the considered uplink resource block.\\ \hline
    $h_T$; $h_H$ & The altitude of TBSs; the altitude of HAPSs.\\ \hline
    $G_m$; $L_N$; $L_F$ & The maximum main lobe gain in $\rm dB$; the near-in-side-lobe level; the far-side-lobe level.\\ \hline
    $p_T^{\rm Tx}$; $p_{H}^{\rm Tx}$; $p_{U}^{\rm Tx}$ & The transmit power of each TBS, HAPS or UE.\\ \hline
    $f_c$; $c$; $N_0$; $B$ & The carrier frequency; the speed of light; the thermal noise power; the bandwidth.\\ \hline
    $W_H^{k}$, $W_T^j$ & The DL channel gain between UE located at $\textbf{z}_{U}^{0}$ and HAPS located at $\textbf{z}_H^{k}$; the DL channel gain between UE located at $\textbf{z}_{U}^{0}$ and TBS located at $\textbf{z}_T^{j}$.\\ \hline
    $W_{U}^{0}$; $W_{U}^{i}$ & The UL channel gain between UE located at $\textbf{z}_U^0$ and its serving HAPS located at $\textbf{z}_H^1$; the UL channel gain between UE located at $\textbf{z}_U^i$ and the HAPS located at $\textbf{z}_H^1$.\\ \hline
    $m_H$; $m_T$ & The shape parameter of the HAPS-UE channels; the shape parameter of the TBS-UE channels.\\ \hline
    $\alpha_H$; $\alpha_T$ & The path loss exponent of HAPS-UE channels; the path loss exponent of TBS-UE channels.\\ \hline
    $\psi_{\rm DL}^{k}$ & The off-boresight angle of the UE located at $\textbf{z}_U^0$ regarding the beam direction of interfering HAPS $\textbf{z}_H^k$.\\ \hline
    $\psi_{\rm UL}^{i}$ & The off-boresight angle of the interfering UE located at $\textbf{z}_U^i$ regarding the beam between the UE $\textbf{z}_U^0$ and its serving HAPS $\textbf{z}_H^1$.\\ \hline
    $\psi_b$ & The half of $3\,{\rm dB}$ beamwidth.\\ \hline
    $\tau_{\rm DL}$; $\tau_{\rm UL}$ & The SINR decoding threshold of downlink services; the SINR decoding threshold of uplink services.\\ \hline
    $\textbf{b}(\textbf{o}_b,r_b)$ & A circular area centered at $\textbf{o}_b$ with radius $r_b$.\\ \hline
    $\textbf{a}(\textbf{o}_a,r_{a,1},r_{a,2})$ & An annular area centered at $\textbf{o}_a$ with inner radius $r_{a,1}$ and outer radius $r_{a,2}$.\\ \hline

     \hline
    \end{tabular}
    }
\end{center}
\label{tab:TableOfNotations}
%\vspace{-8mm}
\end{table*}

Inspired by realistic cases with terrain-constrained TNs, such as the Eastern desert in Egypt, the Mediterranean, and Borneo, we assume that TBSs are deployed in an annular area $\textbf{a}(\textbf{o}_T,r_{\rm in},r_{\rm out})$ centered at $\textbf{o}_T=(0,0,h_T)$ with an inner radius $r_{\rm in}$ and an outer radius $r_{\rm out}$, along rivers, transportation backbones, or the coastlines. This abstraction of the hard-to-reach area preserves both the analytical tractability of distance distributions and coverage performance, while retaining representative features of real-world scenarios \cite{10050345}. Note that we adopt polar coordinates in this paper. The locations of TBSs are assumed to follow a PPP $\Phi_T=\{\textbf{z}_T^j\}$ with density $\lambda_T$, where $\textbf{z}_T^j$ represents the location of the $j$th TBS. To fill the coverage hole inside this circular area, we propose to deploy $N_H$ HAPSs in the circular area $\textbf{b}(\textbf{o}_H,r_{\rm in})$ centered at $\textbf{o}_H=(0,0,h_H)$ with radius $r_{\rm in}$. We assume that the locations of HAPSs follow a BPP $\Phi_H=\{\textbf{z}_H^k\}$, where $\textbf{z}_H^k$ represents the location of the $k$th HAPS. The locations of interfering grant-free active UEs transmitting on the same uplink resource block are modeled as a PPP $\Phi_U=\{\textbf{z}_U^{i}\}$ with a density $\lambda_U$ in the hard-to-reach area $\textbf{b}(\textbf{o}_U,r_{\rm in})$ where $\textbf{o}_U=(0,0,0)$, while a scheduled (grant-based) typical UE is located at $\textbf{z}_U^0=(r_U,\theta_U,0)$.

To better capture the characteristics of real-world beamforming antenna gain of HAPS communications, we adopt the ITU antenna pattern introduced in \cite{10753096,9178753,recommendation2000minimum}. We let $\psi$ denote the off-boresight angle in radians and let $\psi_d=(180/\pi)\psi$ denote the corresponding angle in degrees. The antenna gain in dB at the off-boresight angle $\psi$ is expressed as:
\begin{equation}
    G_{\rm ITU}^{\rm dB}(\psi)=\left\{
\begin{array}{r@{}lc}
& G_m - 3(\frac{\psi_d}{\psi_b})^2, & 0^\circ\leq \psi_d\leq \psi_1\\
& G_m + L_N, &\psi_1<\psi_d\leq\psi_2\\
& G_m+L_N-60\log_{10} (\frac{\psi_d}{\psi_2}),& \psi_2<\psi_d\leq \psi_3\\
& L_F, &\psi_d>\psi_3,
\end{array}
    \right.
\end{equation}
where $G_m$ is the maximum gain in the main lobe in dB, $\psi_b$ is one-half the 3 {\rm dB} beamwidth, $L_N$ is the near-in-side-lobe level, $L_F=G_m-73\,{\rm dBi}$ is the far-side-lobe level. Other parameters can be calculated using:
\begin{equation}
    \psi_1=\psi_b\sqrt{-L_N/3},\,\psi_2=3.745\psi_b,
\end{equation}
% \begin{equation}
%     \psi_2=3.745\psi_b,
% \end{equation}
% \begin{equation}
%     X=G_m+L_N+60\log (\psi_2)
% \end{equation}
and
\begin{equation}
    \psi_3=\psi_2 10^{\frac{G_m+L_N-L_F}{60}}.
\end{equation}
The relationship between the 3 dB beamwidth and the maximum main lobe gain is:
\begin{equation}
    G_m=10\log_{10}(\frac{80^2}{(2\psi_b)^2})+6.
\end{equation}

Using the ITU HAPS antenna pattern, we introduce the signal propagation model for downlink and uplink transmissions of the considered HAPS-based solution in hard-to-reach areas.

\subsection{Downlink Transmission}
In downlink analysis, we assume that the transmitted power of each HAPS is $p_H^{\rm Tx}$ and the transmitted power of each TBS is $p_T^{\rm Tx}$. The transmit gains of each HAPS and each TBS are $G_H^{\rm Tx}$ and $G_T^{\rm Tx}$, respectively, and the receiver gain is $G_U^{\rm Rx}$. The carrier frequency is $f_c$ with wavelength $\lambda$, and $c$ represents the speed of light. We denote $s_A$ as a constant factor accounting for atmospheric absorption and environmental attenuation.  Therefore, the average received signal power at a unit distance transmitted by HAPS is:
\begin{equation}
    p_{H}=p_H^{\rm Tx}G_H^{\rm Tx}G_U^{\rm Rx}s_A(\frac{c}{4\pi f_c})^2 ,
\end{equation}
and the average received signal power at a unit distance transmitted by TBS is
\begin{equation}
    p_T=p_T^{\rm Tx}G_T^{\rm Tx}G_U^{\rm Rx}s_A(\frac{c}{4\pi f_c})^2.
\end{equation}
The channel gains between the HAPSs and the UE are denoted using $\{W_H^k\}$, and the channel gains between the TBSs and the UE are denoted as $\{W_T^j\}$. The signal amplitudes follow the Nakagami-m distribution with parameters $m_H$ and $m_T$, respectively, so that the channel gains follow the Gamma distribution with the corresponding parameters and unit expectations \cite{11288862}. As commonly adopted in large-scale network analysis for analytical tractability, we model TBSs with omnidirectional antennas and focus on comparing the cases where HAPSs adopt either omnidirectional or directional antennas. This assumption provides a tractable baseline for evaluating the impact of HAPS directional beamforming. Without considering the effect of directional beamforming and using the omnidirectional antennas assumption, the signal power transmitted from the HAPS located at $\textbf{z}_H^k$ and received by the UE $\textbf{z}_U^0$ is $p_H W_H^k {d_H^k}^{-\alpha_H}=p_H W_H^k\|\textbf{z}_U^0-\textbf{z}_H^k\|^{-\alpha_H}$. Also, the signal power transmitted from the TBS located at $\textbf{z}_T^j$ and received by the UE $\textbf{z}_U^0$ is $p_T W_T^j {d_T^j}^{-\alpha_T}=p_T W_T^j\|\textbf{z}_U^0-\textbf{z}_T^j\|^{-\alpha_T}$, where $\alpha_H$ and $\alpha_T$ represent the path loss exponent for HAPSs and TBSs, respectively. 

Further considering the ITU HAPS antenna pattern, we define $\overrightarrow{\textbf{z}_H^k\textbf{z}_U^0}$ as the vector from an interfering HAPS located at $\textbf{z}_H^{k}$ to the UE located at $\textbf{z}_U^0$, and the off-boresight angle between the beam direction of HAPS $\textbf{z}_H^{k}$ and $\overrightarrow{\textbf{z}_H^k\textbf{z}_U^0}$ is $\psi_{\rm DL}^{k}$. For tractability, the off-boresight angles of interfering HAPS beams with respect to the considered UE are modeled as uniformly distributed over $[0,\pi]$. The antenna gain compared to the maximum main lobe gain can be calculated using $G_{\rm DL}^{k}=10^{\frac{G_{\rm ITU}^{\rm dB}(\psi_{\rm DL}^{k})-G_{\rm ITU}^{\rm dB}(0)}{10}}$. When the UE is associated with the closest HAPS located at $\textbf{z}_H^{1}$, $\psi_{\rm DL}^{1}=0$ and the relative gain is $G_{\rm DL}^{1}=1$. According to this, the downlink SINR is formulated as:
\begin{equation}
\begin{array}{r@{}l}
    &{\rm SINR}_{\rm DL}^{H}\\
    &=\frac{\displaystyle p_HW_H^1 {d_H^1}^{-\alpha_H}}{\displaystyle\sum_{k\neq 1}p_H W_H^k G_{\rm DL}^{k} {d_H^k}^{-\alpha_H}+\sum_{j}p_T  W_T^j {d_T^j}^{-\alpha_T}+N_0}.
\end{array}
\end{equation}

Also, when the UE is connected to the closest TBS located at $\textbf{z}_T^{1}$, the SINR is formulated as:
\begin{equation}
\begin{array}{r@{}l}
    &{\rm SINR}_{\rm DL}^T\\
    &=\frac{\displaystyle p_TW_T^1 {d_T^1}^{-\alpha_T}}{\displaystyle\sum_{k}p_H W_H^k G_{\rm DL}^{k} {d_H^k}^{-\alpha_H}+\sum_{j\neq 1}p_T W_T^j {d_T^j}^{-\alpha_T}+N_0}.
\end{array}
\label{SINRDLB}
\end{equation}

We assume that the UE connects to the access point (AP) with the strongest average power, which is widely adopted in system-level analysis addressing long-term stable accesses. Specifically, for the considered UE, we assume that the serving beam of its associated HAPS is aligned with the UE during the corresponding service time slot, so that the serving HAPS link can obtain the main-lobe antenna gain. Under this HAPS beam-alignment assumption, the considered UE can be associated with its nearest HAPS or TBS, since the desired average signal strength is mainly dominated by the propagation distance. Let $Q$ denote the associated tier, which satisfies:
\begin{equation}
    Q=\left\{\begin{array}{@{}ll}
        T, & p_H {d_H^1}^{-\alpha_H}\leq p_T {d_T^1}^{-\alpha_T},\\
        H, & p_H {d_H^1}^{-\alpha_H}> p_T {d_T^1}^{-\alpha_T}.
    \end{array}\right.
\end{equation}
Therefore, the downlink coverage probability of the UE located at $\textbf{z}_U^0$ can be formulated as:
\begin{equation}
\begin{array}{r@{}l}
    P_{\rm DL}(\tau_{\rm DL})&=\mathbb{P}\{ {\rm SINR}_{\rm DL}^T>\tau_{\rm DL},Q=T\}\\
    &\;\;\;\;\;\;\;\;\;\;\;\;\;\;+\mathbb{P}\{{\rm SINR}_{\rm DL}^{H}>\tau_{\rm DL},Q=H\}.
\end{array}
\label{PDL_tauDL_fomulation}
\end{equation}

\subsection{Uplink Transmission}
In the hard-to-reach areas, HAPSs can provide not only downlink services, but also uplink services. We assume that the uplink services operate in a different frequency band from the downlink services. For instance, IoT gateways (GWs) collect the information from neighboring IoT devices, and then upload it to HAPSs. Mobile users can transmit information to HAPSs through direct ground-to-air (G2A) links. In this paper, we focus on the uplink performance of a scheduled UE associated with its closest HAPS as a baseline case for direct user-to-HAPS transmissions. The association between the typical UE and its serving HAPS is assumed to have been established before the uplink transmission. We further assume that the serving HAPS has acquired the direction of the scheduled UE during the scheduling procedure. Hence, during the corresponding uplink service time slot, the serving HAPS aligns its receive beam toward the typical UE. Meanwhile, other grant-free active UEs independently transmit on the same resource block. Their signals are received according to the HAPS receive antenna pattern with the corresponding off-boresight gains and are treated as co-channel interference. We model the locations of these grant-free interfering UEs by a PPP with a density $\lambda_U$. The transmitted powers of the scheduled typical UE and grant-free interfering UEs are assumed to be $p_{U}^{\rm Tx}$. The transmit gain is $G_{U}^{\rm Tx}$ and the receiver gain of the HAPS is $G_{H}^{\rm Rx}$. Therefore, without considering the directional beamforming, the signal power transmitted by the UE and received at a unit distance by the HAPS is:
\begin{equation}
    p_{U}=p_U^{\rm Tx}G_U^{\rm Tx}G_H^{\rm Rx}s_A(\frac{c}{4\pi f_c})^2.
\end{equation}

We denote the channel gain between the serving HAPS and the interfering UE located at $\textbf{z}_U^i$ as $\{W_U^i\}$. The signal amplitudes also follow the Nakagami-m distribution with parameter $m_H$, and the channel gains follow the Gamma distribution with parameter $m_H$ and unit expectations. In the uplink procedure, we assume that the typical UE located at $\textbf{z}_U^0$ is served by its closest HAPS located at $\textbf{z}_H^{1}$ with the smallest path loss among HAPSs. We adopt the ITU HAPS antenna pattern and assume that the serving beam of this HAPS is aligned with the typical UE during the corresponding service time slot. We define the beam direction toward the typical UE as $\overrightarrow{\textbf{z}_H^1\textbf{z}_U^{0}}$. Because the uplink services operate in different frequency bands from the downlink services, the uplink signal from the considered UE is interfered by grant-free UEs transmitting on the same resource block. We define $\overrightarrow{\textbf{z}_H^1\textbf{z}_U^{i}}$ as the vector from the HAPS located at $\textbf{z}_H^{1}$ and any interfering UE located at $\textbf{z}_U^i$. The angle between the beam direction of HAPS $\textbf{z}_H^{1}$ and $\overrightarrow{\textbf{z}_H^1\textbf{z}_U^{i}}$ is $\psi_{\rm UL}^{i}$, which is also called the off-boresight angle to the main lobe direction of the considered HAPS. We assume that the antenna gain compared to the maximum main lobe gain is $G_{\rm UL}^{i}=10^{\frac{G_{\rm ITU}^{\rm dB}(\psi_{\rm UL}^{i})-G_{\rm ITU}^{\rm dB}(0)}{10}}$. Therefore, the signal power transmitted from UE located at $\textbf{z}_U^{i}$ and received by the HAPS located at $\textbf{z}_H^{1}$ is
$p_U W_U^i G_{\rm UL}^{i}{d_U^i}^{-\alpha_H}=p_U W_U^i G_{\rm UL}^{i}\|\textbf{z}_U^i-\textbf{z}_H^1\|^{-\alpha_H}$. Based on this, the uplink SINR is formulated as:

\begin{equation}
    {\rm SINR}_{\rm UL}^{H}=\frac{p_U W_U^0 {d_U^0}^{-\alpha_H}}{\displaystyle\sum_{\textbf{z}_U^i\in\Phi_U}p_U W_U^i G_{\rm UL}^{i}{d_U^i}^{-\alpha_H}+N_0}.
\end{equation}
Therefore, the uplink coverage probability of the UE located at $\textbf{z}_U^0$ can be formulated as:
\begin{equation}
    P_{\rm UL}^H(\tau_{\rm UL})=\mathbb{P}\{{\rm SINR}_{\rm UL}^{H}>\tau_{\rm UL},Q=H\}.
\end{equation}
However, to obtain the expressions of downlink coverage probability and uplink coverage probability, we start by deriving the distance distributions in the following section. The definitions of main notations used in this paper are summarized in Table \ref{tab:TableOfNotations}.

\section{Distance Distributions}
In this paper, we focus on the performance analysis of a HAPS network in hard-to-reach areas. In the downlink case, the UE associates with its closest HAPS or its closest TBS, which provides the strongest average received power. In the uplink case, we focus on the decoding performance of HAPSs where the UE is associated with its closest HAPS. Therefore, it is necessary to first derive the distributions of the distances from the considered UE to the HAPSs and TBSs. 
\subsection{Distance to HAPSs}
We assume that the location of the typical UE is $\textbf{z}_U^0=(r_U,\theta_U,0)$, with $\theta_U=0$ and $r_U\leq r_{\rm in}$. The minimum distance between the UE and HAPSs is 
\begin{equation}
\begin{array}{r@{}l}
    d_{H}^{\rm min}&=h_H,
\end{array}
\end{equation}
and the maximum distance between the UE and HAPSs is
\begin{equation}
    d_{H}^{\rm max}=\sqrt{h_H^2+(r_U+r_{\rm in})^2},
\end{equation}
where $\mathbbm{1}(A)$ denotes the indicator function, which equals 1 if event $A$ occurs and 0 otherwise. We define the distance between the UE and its closest HAPS as $D_H$, and the projection of $\textbf{z}_U^0$ on the HAPS plane is $\textbf{z}_U^{H}=(r_U,0,h_H)$. The area of the entire HAPS deployment region $\textbf{b}(\textbf{o}_H,r_{\rm in})$ is $\mathcal{S}_H=|\textbf{b}(\textbf{o}_H,r_{\rm in})|=\pi r_{\rm in}^2$. Let $\mathcal{B}_{H}(d_H)$ denote the area of the HAPS deployment region lying within a distance $d_{H}$ to the considered UE, which is the intersection of two circular areas $\textbf{b}(\textbf{o}_H,r_{\rm in})$ and $\textbf{b}(\textbf{z}_U^H,\sqrt{d_H^2-h_H^2})$. Therefore, $\mathcal{B}_{H}(d_H)=|\textbf{b}(\textbf{o}_H,r_{\rm in})\cap \textbf{b}(\textbf{z}_U^H,\sqrt{d_H^2-h_H^2})|$. Based on this, we introduce the distribution of the distance from the UE to its closest HAPS in Lemma \ref{lem:d_H}.

\begin{lemma}
    The complementary cumulative distribution function (CCDF) of $D_H$ is expressed as:
    \begin{equation}
        \overline{F}_{D_H}(d_H)=\bigg(1-\frac{\mathcal{B}_H(d_H)}{\mathcal{S}_H}\bigg)^{N_H},
    \label{FDHdH}
    \end{equation}
    and the probability density function (PDF) of $D_H$ is expressed as:
    \begin{equation}
        f_{D_H}(d_H)=2N_H\frac{\theta_H(d_H)}{\mathcal{S}_H} d_H\bigg(1-\frac{\mathcal{B}_H(d_H)}{\mathcal{S}_H}\bigg)^{N_H-1},
    \end{equation}
    where 
    \begin{equation}
    \begin{array}{r@{}l}
        \mathcal{B}_H(d_H)&=\theta_H(d_H)(d_H^2-h_H^2)\\
        &\;\;\;\;\;\;\;\;+\phi_H(d_H)r_{\rm in}^2-\sin \phi_H(d_H)r_{\rm in}r_U,
    \end{array}
    \end{equation}
    with 
    \begin{equation}
        \theta_H(d_H)=\left\{
        \begin{array}{l}
        \pi, \;\;r_{\rm in}> r_U, 0<\sqrt{d_H^2-h_H^2}\leq r_{\rm in}- r_U\\
        \arccos \frac{r_U^2+d_H^2-h_H^2-r_{\rm in}^2}{2r_U \sqrt{d_H^2-h_H^2}}, \\
        \;\;\;\|r_U-r_{\rm in}\|\leq \sqrt{d_H^2-h_H^2} \leq\|r_U+r_{\rm in}\|  \\
        0, \;\;{\rm otherwise},
        \end{array}
        \right.  
        \label{thetaH}
    \end{equation}
    \begin{equation}
        \phi_H(d_H)=\left\{
        \begin{array}{l}
        \pi, \;\;\sqrt{d_H^2-h_H^2}>r_U+r_{\rm in}\\
        \arccos\frac{r_U^2+r_{\rm in}^2-d_H^2+h_H^2}{2r_U r_{\rm in}}, \\
        \;\; \|r_U-r_{\rm in}\|\leq \sqrt{d_H^2-h_H^2} \leq\|r_U+r_{\rm in}\|\\
        0, \;\;{\rm otherwise}.
        \end{array}
        \right.
        \label{phiH}
    \end{equation}
\label{lem:d_H}
\end{lemma}
\begin{IEEEproof}
     The CCDF of $D_H$ is given by the probability that all HAPSs are located farther than $d_H$ from the UE. Since the locations of the HAPSs are assumed to be independent, $\overline{F}_{D_H}(d_H)$ can be expressed as the product of the individual probabilities. The PDF of $D_H$ can be obtained by differentiating $1-\overline{F}_{D_H}(d_H)$ with respect to $d_H$.
\end{IEEEproof}
Next, we further investigate the distribution of the distance from the UE to its closest TBS. This is necessary for both cases when UEs connect to TBSs or when they
connect to HAPS but are still interfered with by TBSs.

\subsection{Distance to the closest TBS}
Recall that the location of the typical UE is $\textbf{z}_U^0=(r_U,\theta_U,0)$ where $\theta_U=0$. The minimum distance between the UE and TBSs is 
\begin{equation}
    \begin{array}{r@{}l}
    d_T^{\min}&=\sqrt{h_T^2+(r_U-r_{\rm in})^2},
    \end{array}
\end{equation}
and the maximum distance between the UE and TBSs is
\begin{equation}
    d_T^{\rm max}=\sqrt{h_T^2+(r_U+r_{\rm out})^2}.
\end{equation}
We define the distance between the UE and its closest TBS as $D_T$, and the projection of $\textbf{z}_U^0$ on the TBS plane is $\textbf{z}_U^T=(r_U,0,h_T)$. The area of the entire TBS deployment area $\textbf{a}(\textbf{o}_T,r_{\rm in},r_{\rm out})$ is $\mathcal{S}_T=|\textbf{a}(\textbf{o}_T,r_{\rm in},r_{\rm out})|=\pi r_{\rm out}^2-\pi r_{\rm in}^2$. Let $\mathcal{B}_T(d_T)$ denote the area of the TBS deployment region lying within a distance $d_T$ to the considered UE, which is the intersection of the annular area $\textbf{a}(\textbf{o}_T,r_{\rm in},r_{\rm out})$ and the circular area $\textbf{b}(\textbf{z}_U^T,\sqrt{d_T^2-h_T^2})$. Therefore, $\mathcal{B}_T(d_T)=|\textbf{a}(\textbf{o}_T,r_{\rm in},r_{\rm out})\cap \textbf{b}(\textbf{z}_U^T,\sqrt{d_T^2-h_T^2})|$. Based on this, we introduce the distribution of the distance from the UE to its closest TBS in Lemma \ref{lem:d_T}.

\begin{lemma}
    The CCDF of $D_T$ is expressed as:
\begin{equation}
    \overline{F}_{D_T}(d_T)=\exp(-\lambda_T \mathcal{B}_T(d_T)),
    \label{FDBdB}
\end{equation}
and the PDF of $D_T$ is expressed as
\begin{equation}
\begin{array}{r@{}l}
    f_{D_T}(d_T)&=2\lambda_T \theta_T(d_T) d_T\exp(-\lambda_T \mathcal{B}_T(d_T))\\
    &\;\;\;\;\;\;\;\;\;\;\;\;\;\;\;\;\;\;\;+\delta_{\infty}(d_T)\exp(-\lambda_T \mathcal{S}_T),
\end{array}
\end{equation}
where $\delta_{\infty}(d_T)$ denotes a Dirac point mass at $+\infty$, $\mathcal{B}_T(d_T)=\mathcal{B}_T^{\rm out}(d_T)-\mathcal{B}_T^{\rm in}(d_T)$ and $\theta_T(d_T)=\theta_T^{\rm out}(d_T)-\theta_T^{\rm in}(d_T)$. $\mathcal{B}_T^{\rm in}(d_T)=|\textbf{b}(\textbf{o}_T,r_{\rm in})\cap \textbf{b}(\textbf{z}_U^T,\sqrt{d_T^2-h_T^2})|$ represents the intersection area of the circular area $\textbf{b}(\textbf{o}_T,r_{\rm in})$ and the circular area $\textbf{b}(\textbf{z}_U^T,\sqrt{d_T^2-h_T^2})$, which can be expressed as:
\begin{equation}
    \begin{array}{r@{}l}
    \mathcal{B}_T^{\rm in}(d_T)&=\theta_T^{\rm in}(d_T)(d_T^2-h_T^2)\\
    &\;\;\;\;\;\;\;\;\;\;\;+\phi_T^{\rm in}(d_T)r_{\rm in}^2-\sin \phi_T^{\rm in}(d_T)r_{\rm in}r_U,
    \end{array}
    \label{BBin}
\end{equation}
where the angles $\theta_T^{\rm in}(d_T)$ and $\phi_T^{\rm in}(d_T)$ are defined as:
    \begin{equation}
        \theta_T^{\rm in}(d_T)=\left\{
        \begin{array}{l}
        \pi, \;\;r_{\rm in}> r_U, 0<\sqrt{d_T^2-h_T^2}\leq r_{\rm in}- r_U\\
        \arccos \frac{r_U^2+d_T^2-h_T^2-r_{\rm in}^2}{2r_U \sqrt{d_T^2-h_T^2}}, \\
        \;\;\;\|r_U-r_{\rm in}\|\leq \sqrt{d_T^2-h_T^2} \leq\|r_U+r_{\rm in}\|  \\
        0, \;\;{\rm otherwise},
        \end{array}
        \right.  
        \label{thetaBin}
    \end{equation}
    \begin{equation}
        \phi_T^{\rm in}(d_T)=\left\{
        \begin{array}{l}
        \pi, \;\;\sqrt{d_T^2-h_T^2}>r_U+r_{\rm in}\\
        \arccos\frac{r_U^2+r_{\rm in}^2-d_T^2+h_T^2}{2r_U r_{\rm in}}, \\
        \;\; \|r_U-r_{\rm in}\|\leq \sqrt{d_T^2-h_T^2} \leq\|r_U+r_{\rm in}\|\\
        0, \;\;{\rm otherwise}.
        \end{array}
        \right.
        \label{phiBin}
    \end{equation}
    Similarly, $\theta_T^{\rm out}(d_T)$ and $\phi_T^{\rm out}(d_T)$ are formulated by replacing $r_{\rm in}$ in \eqref{thetaBin} and \eqref{phiBin} with $r_{\rm out}$, respectively. Also, $\mathcal{B}_T^{\rm out}(d_T)$ can be obtained by replacing $\theta_T^{\rm in}(d_T)$, $\phi_T^{\rm in}(d_T)$ and $r_{\rm in}$ in \eqref{BBin} with $\theta_T^{\rm out}(d_T)$, $\phi_T^{\rm out}(d_T)$ and $r_{\rm out}$, respectively.
\label{lem:d_T}
\end{lemma}

Based on the distance distributions, we can proceed to conduct the downlink and uplink analysis for the UEs in hard-to-reach areas.

\section{Coverage Performance Analysis}
\subsection{Downlink Analysis}
In \eqref{PDL_tauDL_fomulation}, the probability of successful downlink coverage is defined as the sum of $P_{\rm DL}^T(\tau_{\rm DL})=\mathbb{P}\{ {\rm SINR}_{\rm DL}^T>\tau_{\rm DL},Q=T\}$ and $P_{\rm DL}^{H}(\tau_{\rm DL})=\mathbb{P}\{ {\rm SINR}_{\rm DL}^{H}>\tau_{\rm DL},Q=H\}$. Therefore, we first derive the DL coverage probability $P_{\rm DL}^T(\tau_{\rm DL}|d_T)$ of UE located at $\textbf{z}_U^0=(r_U,0,0)$ conditioned on the distance to the serving TBS $d_T$ in Theorem \ref{theo:DLCP_dB}.

\begin{theorem}
    Conditioned on the distance to the associated TBS $d_T$, the DL coverage probability  $P_{\rm DL}^T(\tau_{\rm DL}|d_T)$ is calculated using
    \begin{equation}
    \begin{array}{r@{}l}
         P&\displaystyle_{\rm DL}^T(\tau_{\rm DL}|d_T)=\overline{F}_{D_H}(t_{T2H}(d_T))\\
         &\displaystyle \times \sum_{n=0}^{m_T-1}\frac{s_T^n}{n!}(-1)^{n}\displaystyle\frac{\partial^n}{\partial s_T^n}\mathcal{L}_{N_0+I_{T}^{\backslash 1}+I_{H}}(s_T|d_T),
    \end{array} 
    \end{equation}
where $s_T=\frac{m_T \tau_{\rm DL} d_T^{\alpha_T}}{p_T}$ and $\overline{F}_{D_H}(d_H)$ is shown in \eqref{FDHdH}. $t_{T2H}(d_T)$ is defined as:
    \begin{equation}
        t_{T2H}(d_T)\overset{\triangle}{=}\min \bigg\{\bigg(\frac{p_T}{p_{H}}\bigg)^{-\frac{1}{\alpha_{H}}}d_T^{\frac{\alpha_T}{\alpha_H}},d_{H}^{\rm max}\bigg\}.
    \end{equation}
    $N_0$ represents the thermal noise, $I_{T}^{\backslash 1}$ is the interference from other TBSs, and $I_H$ represents the interference from all HAPSs when the user is associated with its closest TBS. The Laplace transform of $N_0+I_{T}^{\backslash 1}+I_H$ can be calculated through
    \begin{equation}
    \begin{array}{r@{}l}
        &\mathcal{L}_{N_0+I_{T}^{\backslash 1}+I_H}(s_T|d_T)\\
        &=\mathcal{L}_{N_0}(s_T)\mathcal{L}_{I_{T}^{\backslash 1}}(s_T|d_T)\mathcal{L}_{I_H}(s_T|t_{T2H}(d_T)),
    \end{array}
    \end{equation}
    where $\mathcal{L}_{N_0}(s_T)=\exp(-s_T N_0)$,
    \begin{equation}
    \begin{array}{r@{}l}
    &\displaystyle\mathcal{L}_{I_{T}^{\backslash 1}}(s_T|d_T)=\exp\bigg(-\int_{d_T}^{d_T^{\rm max}}\lambda_T 2\theta_T(r) \\
        &\;\;\;\;\;\;\;\;\;\;\;\;\;\;\;\;\;\;\displaystyle\times\bigg[1-\bigg(1+\frac{s_T p_T r^{-\alpha_T}}{m_T}\bigg)^{-m_T}\bigg] r{\rm d}r\bigg),
\end{array}
\label{L_IB1}
\end{equation}
and
\begin{equation}
\begin{array}{r@{}l}
    &\mathcal{L}_{I_H}(s_T|t_{T2H}(d_T))\\
    &\displaystyle=\bigg[\int_{t_{T2H}(d_T)}^{d_H^{\rm max}}\int_{0}^{\pi}\bigg(1+\frac{s_T  p_{H} d^{-\alpha_{H}}}{m_{H}}10^{\frac{G_{\rm ITU}^{\rm dB}(\psi)-G_{\rm ITU}^{\rm dB}(0)}{10}}\bigg)^{-m_H}\\
    &\displaystyle\;\;\;\;\;\;\;\;\;\;\;\;\;\;\;\;\;\;\;\;\times \frac{1}{\pi}f_{D}(d|t_{T2H}(d_T)){\rm d}\psi{\rm d}d\bigg]^{N_H},
\end{array}
\label{L_IH}
\end{equation}
where $D$ denotes the distance from the UE to an arbitrary HAPS conditioned on being farther than $t_{T2H}(d_T)$, and its conditional PDF is
\begin{equation}
    f_{D}(d|t_{T2H}(d_T))=\frac{2\theta_H(d)d}{\mathcal{S}_H-\mathcal{B}_H(t_{T2H}(d_T))}.
\label{modifiedPDF}
\end{equation}

\label{theo:DLCP_dB}
\end{theorem}
\begin{IEEEproof}
    See Appendix \ref{app:DLCP_dB}.
\end{IEEEproof}
When HAPSs are equipped with omnidirectional antennas rather than directional antennas, the Laplace transform of interference is reduced as shown in Remark \ref{rem:1}.

\begin{remark}
When HAPSs are equipped with omnidirectional antennas, the Laplace transform of $I_H$ simplifies to:
\begin{equation}
\begin{array}{r@{}l}
    &\displaystyle\mathcal{L}_{I_H}(s_T|t_{T2H}(d_T))=\bigg[\int_{t_{T2H}(d_T)}^{d_H^{\rm max}}\\
    &\displaystyle \; \times \bigg(1+\frac{s_T  p_{H} d^{-\alpha_{H}}}{m_{H}}\bigg)^{-m_H} f_{D}(d|t_{T2H}(d_T)){\rm d}d\bigg]^{N_H},
\end{array}
\label{L_IH_omnidirectional}
\end{equation}
where $10^{\frac{G_{\rm ITU}^{\rm dB}(\psi)-G_{\rm ITU}^{\rm dB}(0)}{10}}$ in \eqref{L_IH} is replaced by constant 1.
\label{rem:1}
\end{remark}

Next, we derive the DL coverage probability $P_{\rm DL}^{H}(\tau_{\rm DL}|d_H)$ of UE located at $\textbf{z}_U^0$ conditioned on the distance to the serving HAPS $d_H$ in Theorem \ref{theo:DLCP_dH}.

\begin{theorem}
    Conditioned on the distance to the associated HAPS $d_H$, the DL coverage probability $P_{\rm DL}^{H}(\tau_{\rm DL}|d_H)$ is calculated using
    \begin{equation}
    \begin{array}{r@{}l}
         P&\displaystyle_{\rm DL}^{H}(\tau_{\rm DL}|d_H)=\overline{F}_{D_T}(t_{H2T}(d_H))\\
         &\displaystyle \times \sum_{n=0}^{m_H-1}\frac{s_H^n}{n!}(-1)^{n}\displaystyle\frac{\partial^n}{\partial s_H^n}\mathcal{L}_{N_0+I_T+I_{H}^{\backslash 1}}(s_H|d_H),
    \end{array} 
    \end{equation}
    where $s_H=\frac{m_H \tau_{\rm DL} d_H^{\alpha_H}}{p_H}$ and $\overline{F}_{D_T}(d_T)$ is shown in \eqref{FDBdB}. $t_{H2T}(d_{H})$ is defined as:
\begin{equation}
    t_{H2T}(d_{H})\overset{\triangle}{=}\min \bigg\{\bigg(\frac{p_{H}}{p_T}\bigg)^{-\frac{1}{\alpha_T}}d_H^{\frac{\alpha_H}{\alpha_{T}}},d_T^{\rm max}\bigg\}.
\end{equation}
    $I_T$ is the interference from TBSs, and $I_{H}^{\backslash 1}$ represents the interference from other HAPSs. The Laplace transform of $N_0+I_T+I_{H}^{\backslash 1}$ can be calculated through
    \begin{equation}
    \begin{array}{r@{}l}
        &\mathcal{L}_{N_0+I_T+I_{H}^{\backslash 1}}(s_H|d_H)\\
        &=\mathcal{L}_{N_0}(s_H)\mathcal{L}_{I_T}(s_H|t_{H2T}(d_H))\mathcal{L}_{I_{H}^{\backslash 1}}(s_H|d_H),
    \end{array}
    \end{equation}
    where $\mathcal{L}_{N_0}(s_H)=\exp(-s_H N_0)$, $\mathcal{L}_{I_T}(s_H|t_{H2T}(d_H))$ can be obtained by replacing $s_T$ and $d_T$ in \eqref{L_IB1} with $s_H$ and $t_{H2T}(d_H)$. $\mathcal{L}_{I_{H}^{\backslash 1}}(s_H|d_H)$ can be obtained by replacing $s_T$, $t_{T2H}(d_T)$ and $N_H$ in \eqref{L_IH} with $s_H$, $d_H$ and $N_H-1$, respectively. 
\label{theo:DLCP_dH}
\end{theorem}
\begin{IEEEproof}
    Different from Theorem \ref{theo:DLCP_dB}, the interference comes from $N_H-1$ HAPSs rather than $N_H$ HAPSs since the closest HAPS is serving the UE. When the UE is associated with its closest HAPS with distance $d_H$, TBSs are located outside the distance $t_{H2T}(d_H)$.
\end{IEEEproof}

Based on the DL coverage probabilities conditioned on the distances to the closest HAPS or TBS, we can derive the DL coverage probability for the UE in Theorem \ref{theo:DLCP}.

\begin{theorem} The downlink coverage probability of the UE located at $\textbf{z}_U^0=(r_U,0,0)$ is formulated as:
    \begin{equation}
    \begin{array}{r@{}l}
         P_{\rm DL}(\tau_{\rm DL})&=P_{\rm DL}^T(\tau_{\rm DL})+P_{\rm DL}^{H}(\tau_{\rm DL})\\
         &\displaystyle=\int_{h_T}^{d_T^{\rm max}}P_{\rm DL}^T(\tau_{\rm DL}|d_T)f_{D_T}(d_T){\rm d}d_T\\
         &\displaystyle+\int_{h_H}^{d_{H}^{\rm max}}P_{\rm DL}^{H}(\tau_{\rm DL}|d_H)f_{D_H}(d_H){\rm d}d_H,
    \end{array}
    \end{equation}
where the distributions of $d_T$ and $d_H$ are shown in Lemmas \ref{lem:d_T} and \ref{lem:d_H}, respectively, and the conditional DL coverage probabilities $P_{\rm DL}^T(\tau_{\rm DL}|d_T)$ and $P_{\rm DL}^{H}(\tau_{\rm DL}|d_H)$ are shown in Theorems \ref{theo:DLCP_dB} and \ref{theo:DLCP_dH}, respectively.
\label{theo:DLCP}
\end{theorem}
To reduce the computational complexity of conditional coverage probabilities, we adopt the upper bound of the upper incomplete Gamma function shown in \cite{8833522} and obtain their approximations in Remark \ref{rem:2}.
\begin{remark}
    The conditional DL coverage probabilities in Theorems \ref{theo:DLCP_dB} and \ref{theo:DLCP_dH} can be approximated using:
    \begin{equation}
    \begin{array}{r@{}l}
        P&_{\rm DL}^T(\tau_{\rm DL}|d_T)\approx\displaystyle \overline{F}_{D_H}(t_{T2H}(d_T))\\
        &\displaystyle \times\sum_{n=1}^{m_T}\binom{m_T}{n} (-1)^{n+1}\exp(-n\beta_2^T s_T N_0)\\
        &\displaystyle\times\mathcal{L}_{I_{T}^{\backslash 1}}\big(n\beta_2^T s_T|d_T\big)\mathcal{L}_{I_{H}^{}}\big(n\beta_2^T s_T|t_{T2H}(d_T)\big),
    \end{array}
    \end{equation}
and 
\begin{equation}
    \begin{array}{r@{}l}
        P&_{\rm DL}^{H}(\tau_{\rm DL}|d_H)\approx\displaystyle \overline{F}_{D_T}(t_{H2T}(d_H))\\
        &\displaystyle \times \sum_{n=1}^{m_H}\binom{m_H}{n} (-1)^{n+1}\exp(-n\beta_2^H s_H N_0)\\
        &\displaystyle\times\mathcal{L}_{I_{H}^{\backslash 1}}\big(n\beta_2^H s_H|d_H\big)\mathcal{L}_{I_T^{}}\big(n\beta_2^H s_H|t_{H2T}(d_H)\big),
    \end{array}
    \end{equation}
where $\beta_2^T=(m_T!)^{-\frac{1}{m_T}}$ and $\beta_2^H=(m_H!)^{-\frac{1}{m_H}}$.
\label{rem:2}
\end{remark}

\subsection{Uplink Analysis}
In addition to the DL service, it is necessary to evaluate the UL performance via HAPSs. We first investigate the UL coverage probability of the UE located at $\textbf{z}_U^0=(r_U,\theta_U,0)$ conditioned on the location of its closest HAPS. This facilitates the analysis and modeling of interference from other active UEs. On the other hand, when the HAPS locations are fixed rather than random, such conditional coverage probability can directly characterize the UL performance of UEs at different locations in the hard-to-reach area. After that, using the PDF of the location of the closest HAPS, we derive the average UL coverage performance for the case where the locations of HAPSs are random.

We denote the location of the serving HAPS as $\textbf{z}_H^{1}=(r_H^1,\theta_H^1,h_H)$. The unit direction vector from it to the considered UE is $\textbf{v}_0=\overrightarrow{\textbf{z}_H^1\textbf{z}_U^0}/\|\overrightarrow{\textbf{z}_H^1\textbf{z}_U^0}\|$. For any interfering UE located at $\textbf{z}_U^i=(r_U^i,\theta_U^i,0)$, the vector from the serving HAPS to it is $\textbf{v}_i=\overrightarrow{\textbf{z}_H^1\textbf{z}_U^i}/\|\overrightarrow{\textbf{z}_H^1\textbf{z}_U^i}\|$. Therefore, the off-boresight angle of this interfering UE can be calculated using:
\begin{equation}
    \psi_{\rm UL}^{i}=\psi_{\rm UL}(r_U^i,\theta_U^i|\textbf{z}_H^1)=\arccos{(\textbf{v}_0^T \textbf{v}_i)}.
\end{equation}
Therefore, we derive the UL coverage probability of the UE in Theorem \ref{theo:ULCP_oneshot}:

\begin{theorem} Conditioned on the location of the closest HAPS, the UL coverage probability can be calculated using:
\begin{equation}
\begin{array}{r@{}l}
    \displaystyle P_{\rm UL}^{H}&(\tau_{\rm UL}|\textbf{z}_H^1)=\overline{F}_{D_T}(t_{H2T}(\|\textbf{z}_H^1-\textbf{z}_U^0\|))\\
    &\displaystyle\times\sum_{n=0}^{m_H-1}\frac{s_U^n}{n!}(-1)^{n}\displaystyle\frac{\partial^n}{\partial s_U^n}\mathcal{L}_{N_0+I_U}(s_U|\textbf{z}_H^1),
    \end{array} 
    \label{conditionalULcov}
    \end{equation}
    where $s_U=\frac{m_H \tau_{\rm UL} \|\textbf{z}_H^1-\textbf{z}_U^0\|^{\alpha_H}}{p_U}$, $I_U$ represents the interference from grant-free UEs. The Laplace transform of $N_0+I_U$ can be calculated through
    \begin{equation}
    \begin{array}{r@{}l}
        &\mathcal{L}_{N_0+I_U}(s_U|\textbf{z}_H^1)=\mathcal{L}_{N_0}(s_U)\mathcal{L}_{I_U}(s_U|\textbf{z}_H^1),
    \end{array}
    \end{equation}
    where $\mathcal{L}_{N_0}(s_U)=\exp(-s_U N_0)$, and $\mathcal{L}_{I_U}(s_U|\textbf{z}_H^1)$ is calculated using:
\begin{equation}
\begin{array}{r@{}l}
    &\mathcal{L}_{I_U}(s_U|\textbf{z}_H^1){=}\displaystyle \exp\bigg(-\int_{0}^{r_{\rm in}} \int_{-\pi}^{\pi}  \lambda_U \bigg[ 1-\\
    &\displaystyle\;\;\bigg(1+\frac{s_U  p_{U} d(r,\theta|\textbf{z}_H^1)^{-\alpha_{H}}G_{\rm ITU}^{\rm linear}(r,\theta|\textbf{z}_H^1)}{m_{H}} \bigg)^{-m_H}\bigg]r{\rm d}\theta{\rm d}r\bigg), 
\end{array}
\end{equation}
where
\begin{equation}
    d(r,\theta|\textbf{z}_H^1)=\sqrt{h_H^2+(r_H^1)^2+r^2-2r_H^1 r \cos (\theta_H^1-\theta)},
\end{equation}
and
\begin{equation}
    G_{\rm ITU}^{\rm linear}(r,\theta|\textbf{z}_H^1)=10^{\frac{G_{\rm ITU}^{\rm dB}(\psi_{\rm UL}(r,\theta|\textbf{z}_H^1))-G_{\rm ITU}^{\rm dB}(0)}{10}}.
\end{equation}

\label{theo:ULCP_oneshot}
\end{theorem}
\begin{IEEEproof}
    See Appendix \ref{app:ULCP_oneshot}.
\end{IEEEproof}

Next, we derive the PDF of the location of the serving HAPS and evaluate the average UL performance for the UE when the location of the serving HAPS is unknown in Theorem \ref{theo:ULCP}.

\begin{theorem}
    The average UL coverage probability of the UE can be calculated using:
\begin{equation}
    P_{\rm UL}^{H}(\tau_{\rm UL})=\int_{\textbf{b}(\textbf{o}_H,r_{\rm in})}P_{\rm UL}^{H}(\tau_{\rm UL}|\textbf{z}_H^1)f(\textbf{z}_H^1){\rm d}\textbf{z}_H^1,
\end{equation}
where the conditional UL coverage probability is defined in \eqref{conditionalULcov} and $f(\textbf{z}_H^1)$ is formulated as:
\begin{equation}
\begin{array}{r@{}l}
\displaystyle f(\textbf{z}_H^1)=\frac{N_H}{\mathcal{S}_H}&\displaystyle \bigg(1-\frac{\mathcal{B}_{H}\big(\|\textbf{z}_H^1-\textbf{z}_U^0\|\big)}{\mathcal{S}_H}\bigg)^{N_H-1}.
\label{f_zH1}
\end{array}
\end{equation}
\label{theo:ULCP}
\end{theorem}
\begin{IEEEproof}
    The PDF of $\textbf{z}_H^1$ can be obtained directly from order statistics. For a candidate point $\textbf{z}_H^1$, one of HAPSs should be located at this point with probability $1/\mathcal{S}_H$, and the remaining $N_H-1$ HAPSs must be located outside the circular area $\textbf{b}(\textbf{z}_U^H,\sqrt{\|\textbf{z}_H^1-\textbf{z}_U^0\|^2-h_H^2})$, each of them has probability $1-\mathcal{B}_{H}\big(\|\textbf{z}_H^1-\textbf{z}_U^0\|\big)/\mathcal{S}_H$. When the UE is associated with HAPSs, there should be no candidate TBS with a distance less than $t_{H2T}(\|\textbf{z}_H^1-\textbf{z}_U^0\|)$.
\end{IEEEproof}

\section{Simulation Results and Discussion}
In this manuscript, we conduct 20,000 Monte Carlo simulations for DL and UL analysis for a circular hard-to-reach area with radius $200\,{\rm km}$. We assume that HAPSs are uniformly deployed in this area, and TBSs are deployed in the annular area with the same center, whose inner radius and outer radius are $r_{\rm in}=200\,{\rm km}$ and $r_{\rm out}=300\,{\rm km}$. For the DL case, the transmit powers of HAPSs and TBSs are $p_{H}^{\rm Tx}=50\,{\rm dBm}$ and $p_T^{\rm Tx}=43\,{\rm dBm}$ \cite{10040542} and the density of TBSs is assumed to be $\lambda_T=0.1\,{\rm TBSs/km^2}$ \cite{11121572}. The receiver gain of UEs is $3\,{\rm dB}$. The carrier frequency is $2.4\,{\rm GHz}$, and the speed of light is $c=3\times 10^{8}\,{\rm m/s}$. The bandwidth is $B=10\,{\rm MHz}$ and the noise power density is $-174\,{\rm dBm/Hz}$. The shape parameters for the HAPS-UE channels and TBS-UE channels are $m_H=3$ and $m_T=1$, respectively. The path loss exponents of the HAPS-UE channels and TBS-UE channels are $\alpha_H=2.5$ and $\alpha_T=3$, respectively. The altitude of HAPSs is $h_H=20\,{\rm km}$ and the altitude of TBSs is $h_T=30\,{\rm m}$. In the UL case, we further assume that the transmit power of UEs or IoT GWs is $p_U^{\rm Tx}=23\,{\rm dBm}$ and the grant-free interfering UE density in the same resource block is $\lambda_U=0.02\,{\rm UEs/km^2}$ \cite{sendra2020lorawan,reddy2022uplink}. The constant attenuation factor $s_A=0\,{\rm dB}$ represents a favorable channel environment, where the additional atmospheric absorption and environmental attenuation are negligible. The solid lines represent the analytical results and the markers represent the simulation results.
\subsection{Downlink Performance}

In Fig. \ref{fig:DL_without_directional_beamforming}, we show the DL coverage probability versus the distance $r_U$ from a UE to the center of the considered hard-to-reach area when the HAPSs are equipped with omnidirectional antennas. The analysis is shown in Theorem \ref{theo:DLCP} and the Laplace transform of the interference from HAPSs reduces to \eqref{L_IH_omnidirectional} in Remark \ref{rem:1}. The SINR threshold is $\tau_{\rm DL}=-10\,{\rm dB}$. When the number of HAPSs varies from $1$ to $16$, UEs with $r_U=200\,{\rm km}$ mainly connect to TBSs and the interference from HAPSs does not severely affect the connections. However, for UEs with $r_U<200\,{\rm km}$, the DL performance can be improved by adding more HAPSs into the area. When $N_H=1$, the coverage of UE at the center is 0.33, which is much lower than that of UEs covered by TBSs. When $N_H=16$, the UEs inside the hard-to-reach area can connect to a closer HAPS so that the DL coverage probability can reach 0.88. The handover between the TBS network and the HAPS network appears when $r_U$ is close to the boundary of the hard-to-reach area, therefore, the UE with $r_U=180\,{\rm km}$ faces higher TBS interference than the UE at $r_U=0\,{\rm km}$. When $N_H=4$, the DL coverage probability of UEs at $r_U=180\,{\rm km}$ is 0.39, which is much less than that of the UE at $r_U=0\,{\rm km}$. However, when $N_H=16$, the difference between them is reduced from 0.3 to 0.15. Therefore, deploying more HAPSs can help fill the coverage hole and also prevent performance gaps on the edge of the hard-to-reach area. 

\begin{figure}[ht]
    \centering
    \includegraphics[width=0.88\linewidth]{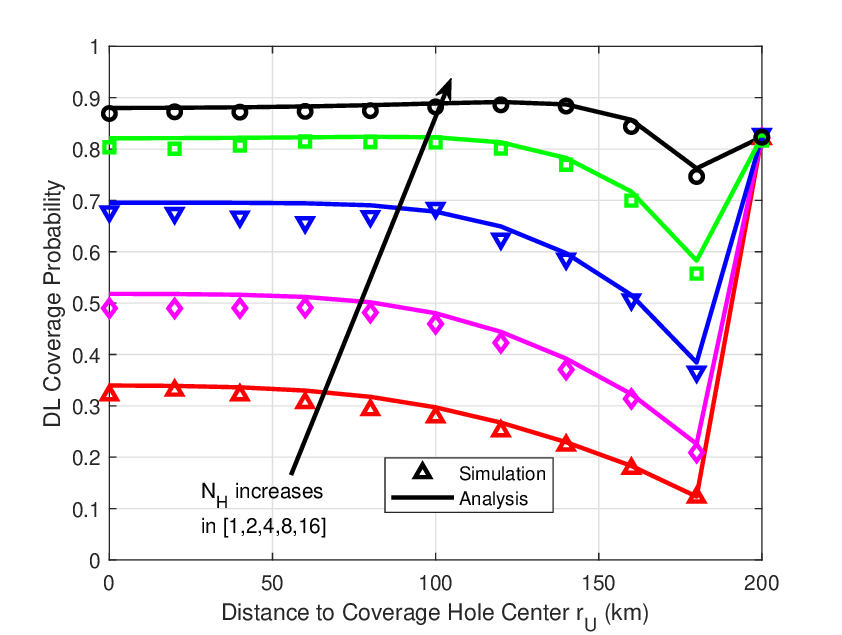}
    \caption{DL performance of UEs versus the distance to the coverage hole center $r_U$, when the HAPSs are equipped with omnidirectional antennas.}
    \label{fig:DL_without_directional_beamforming}
\end{figure}

In Fig. \ref{fig:DL_with_directional_beamforming}, we consider the case where HAPSs are equipped with directional antennas, and show the DL coverage performance in Theorem \ref{theo:DLCP} versus $r_U$, where the half of 3 dB beamwidth is $5^{
\circ}$. Compared to Fig. \ref{fig:DL_without_directional_beamforming}, directional beamforming can further improve DL coverage performance, especially when the number of HAPSs is large. When $N_H=4$, the performance of UE with $r_U=0\,{\rm km}$ can be improved by only $5.89\%$. When $N_H=8$ and $N_H=16$, the improvement can reach $8.41\%$ and $9.82\%$, respectively.

\begin{figure}[ht]
    \centering
    \includegraphics[width=0.88\linewidth]{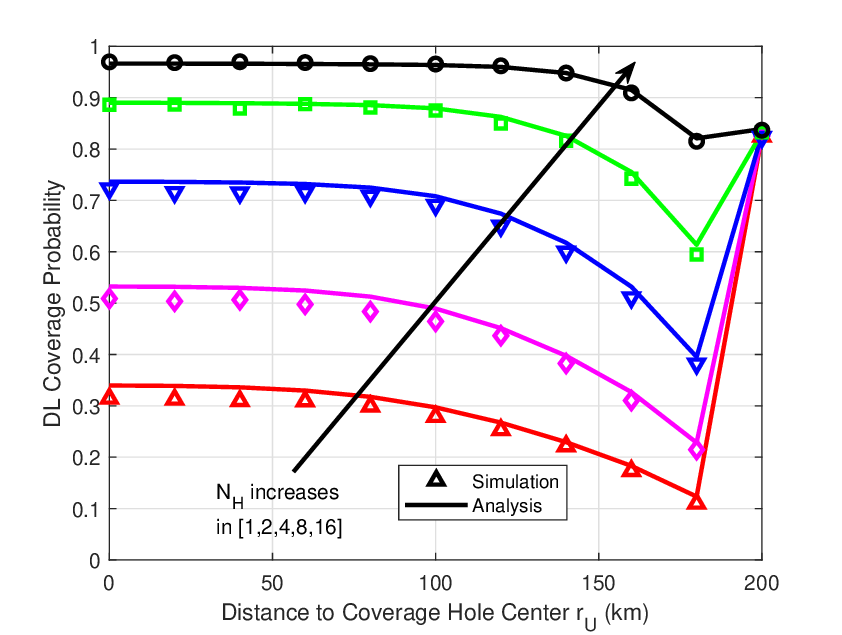}
    \caption{DL performance of UEs versus the distance to the coverage hole center $r_U$, when the HAPSs are equipped with directional antennas and half of 3 dB beamwidth is $5^{\circ}$.}
    \label{fig:DL_with_directional_beamforming}
\end{figure}

In Fig. \ref{fig:DL_Covpro_NH}, we show the relationship between the DL coverage probabilities of the UE at the center of the hard-to-reach area $\textbf{o}_U$ and the number of HAPSs $N_H$, considering different SINR decoding thresholds $\tau_{\rm DL}$ and half of 3 dB beamwidth values $\psi_b$. When $\tau_{\rm DL}=0\,{\rm dB}$ and $\psi_b>20^\circ$, the DL coverage probability cannot exceed 0.5 with fewer than 32 HAPSs. When the beamwidth becomes narrower, e.g. $\psi_b=5^\circ$ or $\psi_b=2.5^\circ$, at least 32 or 26 HAPSs are required to achieve a DL coverage probability of 0.6.
When the SINR decoding threshold decreases to $\tau_{\rm DL}=-10\,{\rm dB}$, the UE can achieve a high coverage probability. To achieve a DL coverage probability of 0.9, the wide-beam case with $\psi_b=80^{\circ}$ requires 18 HAPSs, whereas the narrow-beam case $\psi_b=2.5^\circ$ requires only 8 HAPSs. These results indicate that the required number of HAPSs strongly depends on both the antenna pattern and decoding ability. In practical scenarios, operators should jointly determine the number of HAPSs and antenna pattern design addressing the performance requirement and expense.

\begin{figure}[ht]
    \centering
    \includegraphics[width=0.88\linewidth]{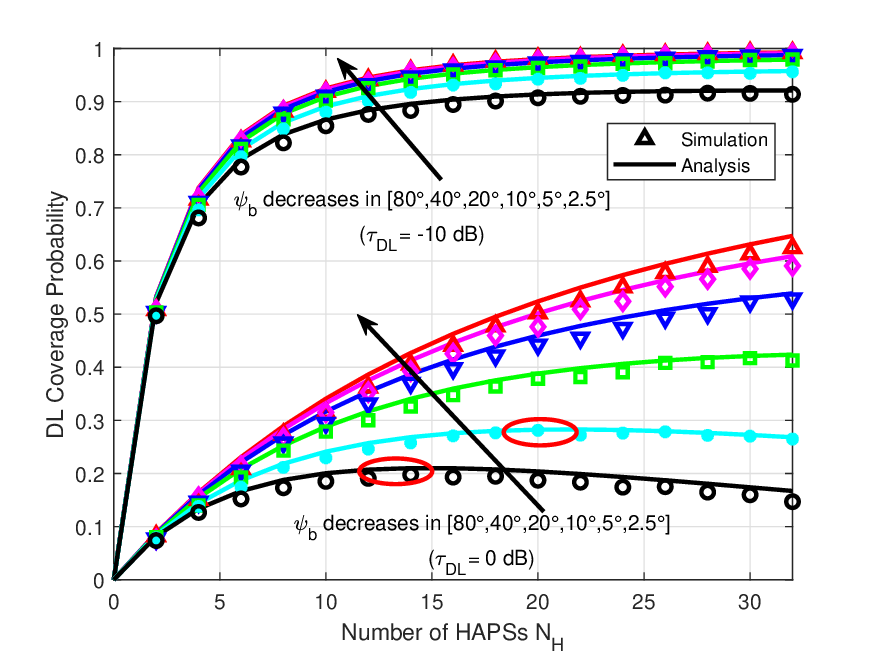}
    \caption{DL performance of the UE at the center of the hard-to-reach area versus the number of HAPSs $N_H$, when $\tau_{\rm DL}=-10\,{\rm dB}$ or $\tau_{\rm DL}=0\,{\rm dB}$.}
    \label{fig:DL_Covpro_NH}
\end{figure}

In Fig. \ref{fig:DL_Covpro_hH}, we further illustrate the relationship between the DL coverage probability and the HAPS altitude when HAPSs are equipped with directional antennas. We consider the case where the number of HAPSs is $N_H=16$. Since the UE is located at the center of the hard-to-reach area, it is far from the TBSs and the received interference mainly comes from the side-lobe signals from other HAPSs. As the HAPS altitude increases, the large-scale path loss becomes more severe, leading to lower received signal strength and reduced coverage probability. In addition to the simulations based on the uniform off-boresight angle approximation, we also conduct simulations based on a random footprint association policy, where each interfering HAPS determines its footprint according to the strongest average received power under main-lobe transmission, and aligns its beam with a randomly selected UE within this footprint. The results show that the gap between the random footprint association policy and uniform off-boresight angle approximation is less than 0.02.

\begin{figure}[ht]
    \centering
    \includegraphics[width=0.88\linewidth]{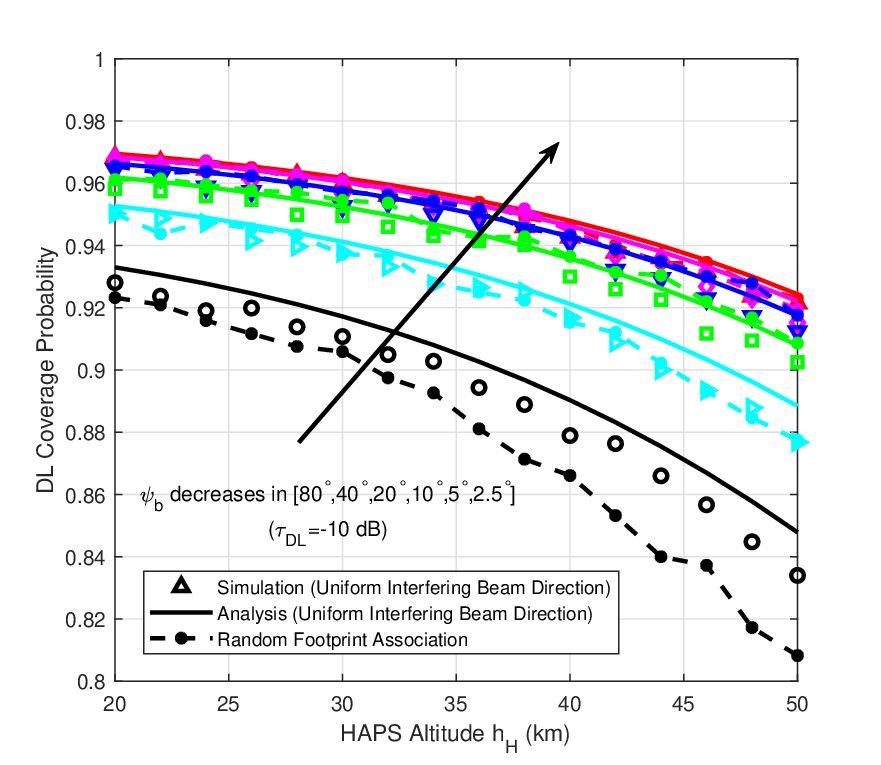}
    \caption{DL performance of the UE at the center of the hard-to-reach area versus the altitude of HAPSs $h_H$.}
    \label{fig:DL_Covpro_hH}
\end{figure}

\subsection{Uplink Performance}

\begin{figure*}[htbp]
\centering
\begin{minipage}{0.32\textwidth}
\centering 
\includegraphics[width=1\linewidth]{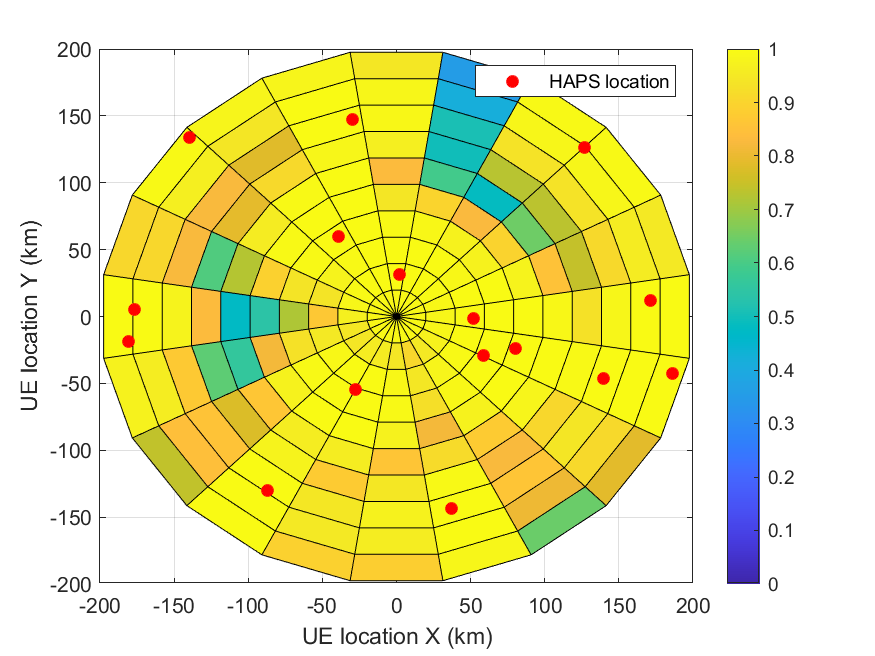}
\caption{Uplink coverage map with HAPS locations (Analysis, $\psi_b=2.5^{\circ}$).}
\label{fig:CovMap_UL_with_HAPS_Location_psi25_Ana}
\end{minipage}
% \end{figure}
\hfill
% \begin{figure}[htbp]
\begin{minipage}{0.32\textwidth}
\centering
\includegraphics[width=1\linewidth]{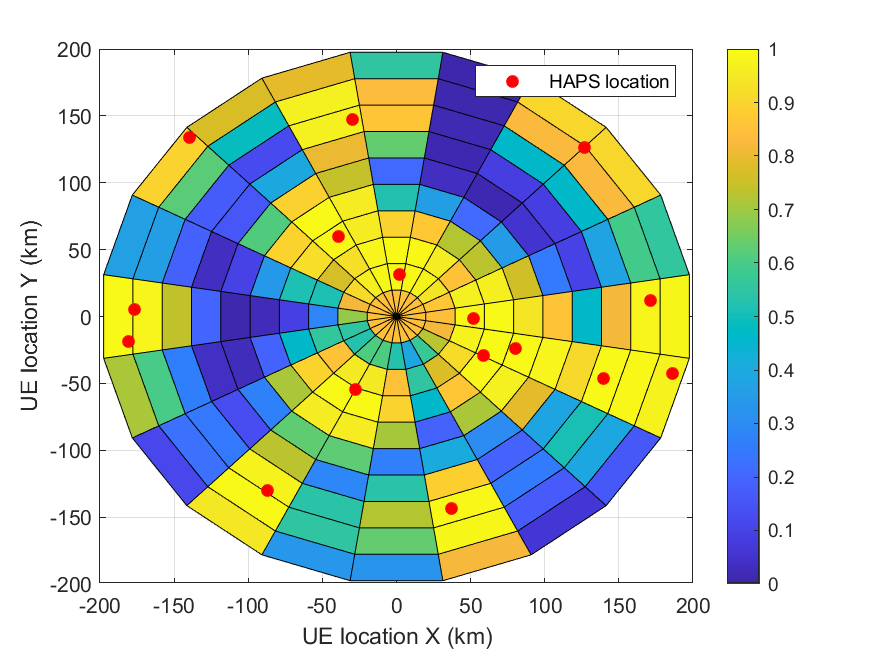}
\caption{Uplink coverage map with HAPS locations (Analysis, $\psi_b=5^{\circ}$).}
\label{fig:CovMap_UL_with_HAPS_Location_psi5_Ana}
\end{minipage}
% \end{figure}
\hfill
% \begin{figure}[htbp]
\begin{minipage}{0.32\textwidth}
\centering
\includegraphics[width=1\linewidth]{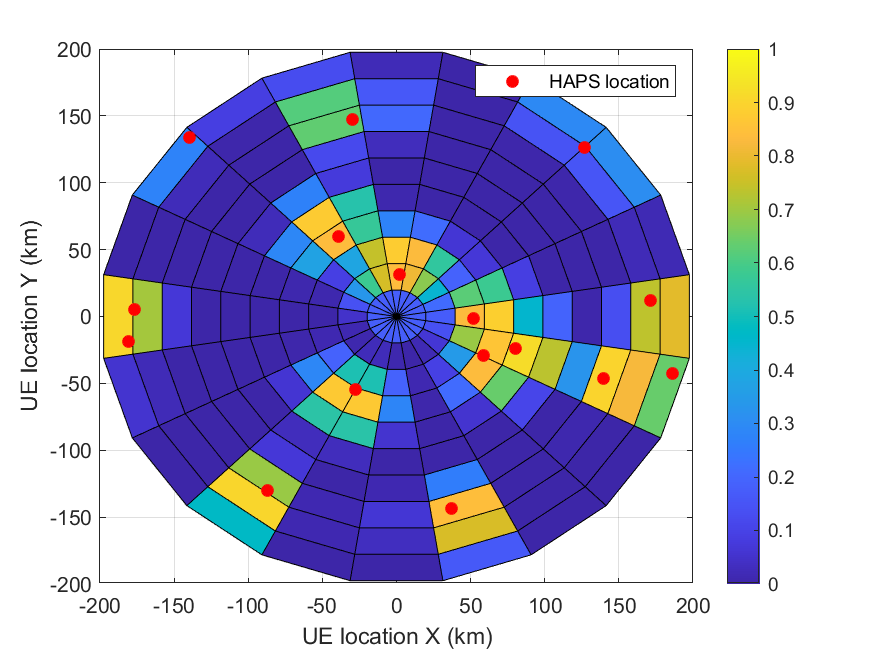}
\caption{Uplink coverage map with HAPS locations (Analysis, $\psi_b=10^{\circ}$).}
\label{fig:CovMap_UL_with_HAPS_Location_psi10_Ana}
\end{minipage}
\end{figure*}

In Theorem \ref{theo:ULCP_oneshot}, we first investigate the UL coverage performance conditioned on the location of the associated HAPS. Based on this, we draw the heatmaps of the UL coverage performance when $\psi_b=2.5^{\circ},5^{\circ},10^{\circ}$ in Figs. \ref{fig:CovMap_UL_with_HAPS_Location_psi25_Ana}, \ref{fig:CovMap_UL_with_HAPS_Location_psi5_Ana}, and \ref{fig:CovMap_UL_with_HAPS_Location_psi10_Ana}, respectively.  We consider $N_H=16$ HAPSs and assume that the SINR decoding threshold is $\tau_{\rm UL}=-10\,{\rm dB}$. When $\psi_b=10^{\circ}$, interference from neighboring UEs cannot be successfully controlled through directional beamforming. Therefore, only the UEs close to the HAPSs can be served with a probability higher than 0.8. When $\psi_b=5^{\circ}$, the UEs surrounding the HAPSs can reach a coverage probability close to 1, but there are still many areas that lack coverage. When $\psi_b=2.5^{\circ}$, most areas can be successfully covered, and only a few locations have coverage probability lower than 0.5 due to an inappropriate HAPS deployment. Therefore, the UL services require HAPSs to be equipped with directional antennas with a 3 dB beamwidth less than $5^{\circ}$. However, this may also lead to beam-switching complexity.

\begin{figure}[ht]
    \centering
    \includegraphics[width=0.88\linewidth]{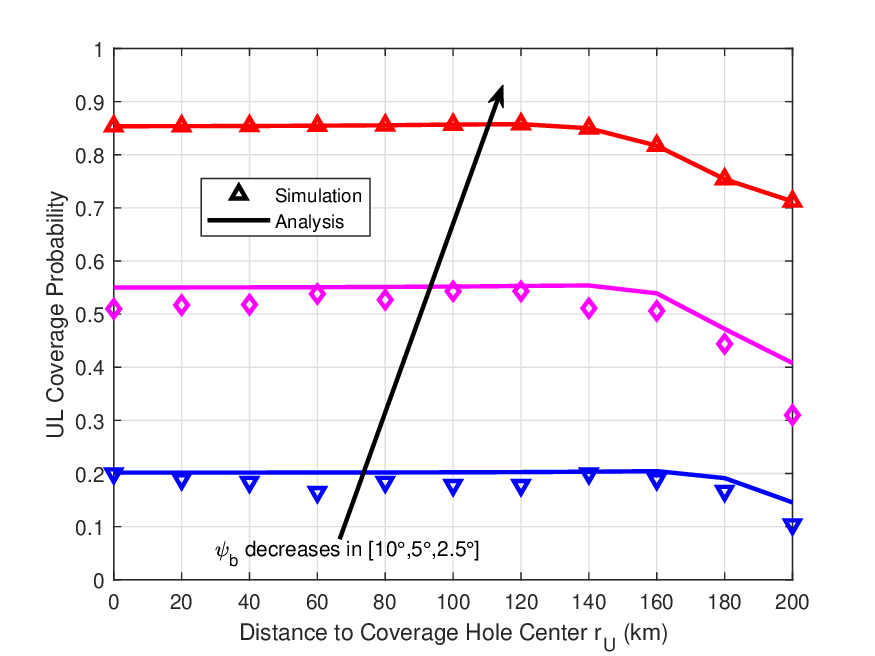}
    \caption{Average UL performance of HAPSs versus the distance to the coverage hole center $r_U$, when the HAPSs are equipped with directional antennas and the number of HAPSs $N_H=16$.}
    \label{fig:UL_with_directional_beamforming_RDM}
\end{figure}

As shown in Fig. \ref{fig:UL_with_directional_beamforming_RDM}, we evaluate the average UL coverage performance introduced in Theorem \ref{theo:ULCP}, where UEs are associated with HAPSs and the locations of HAPSs are random. When $\psi_b=10^{\circ}$, the UL coverage probability of UEs with $r_U<140\,{\rm km}$ is 0.2, which means that only 20\% of UEs within this distance can be successfully covered. However, when $\psi_b=5^{\circ}$ or $\psi_b=2.5^{\circ}$, the UL coverage probabilities can be improved to 0.55 and 0.85, respectively. The results in Fig. \ref{fig:UL_with_directional_beamforming_RDM} correspond to the heatmaps drawn in Figs. \ref{fig:CovMap_UL_with_HAPS_Location_psi25_Ana}, \ref{fig:CovMap_UL_with_HAPS_Location_psi5_Ana}, and \ref{fig:CovMap_UL_with_HAPS_Location_psi10_Ana}. When $r_U>140\,{\rm km}$ and $r_U$ increases, more UEs will connect to TBSs rather than HAPSs. Therefore, the UL coverage probabilities of HAPSs decrease since the association probabilities of HAPSs decrease.

\begin{figure}[ht]
    \centering
    \includegraphics[width=0.88\linewidth]{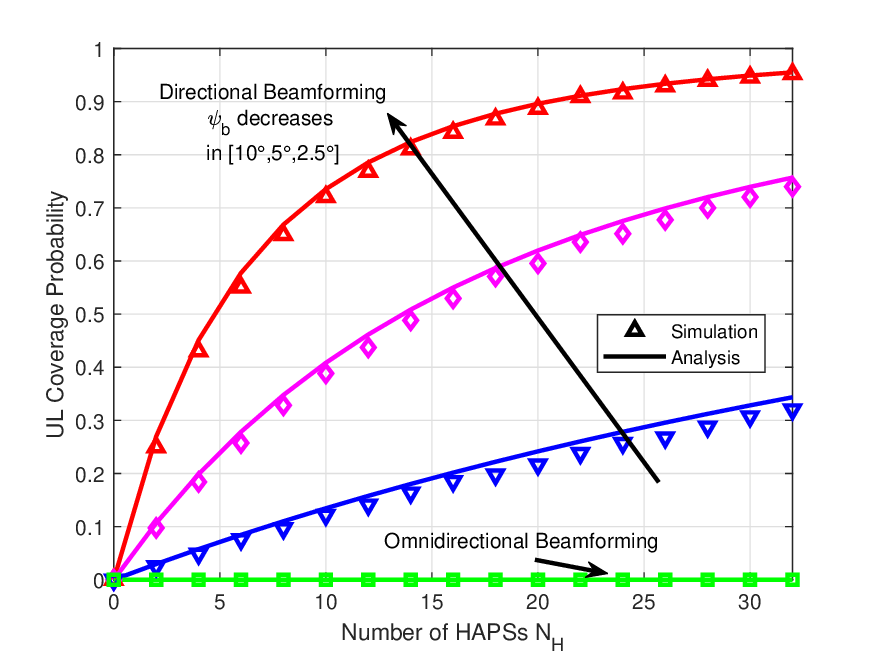}
    \caption{Average UL performance of the UE at $\textbf{o}_U$ versus the number of HAPSs $N_H$, when the HAPSs are equipped with directional antennas or omnidirectional antennas.}
    \label{fig:UL_Covpro_NH}
\end{figure}

In Fig. \ref{fig:UL_Covpro_NH}, we investigate the relationship between the average UL coverage probability of the UE at $\textbf{o}_U$ and the number of HAPSs $N_H$. When HAPSs are equipped with omnidirectional antennas, UL coverage probability is close to 0 due to the high mutual UE interference. When HAPSs are equipped with directional antennas and $\psi_b=10^{\circ}$, UL coverage probability is less than 0.35 when $N_H$ is less than 32. When $\psi_b=5^{\circ}$ or $\psi_b=2.5^{\circ}$, the minimum numbers of HAPSs to achieve a UL coverage probability of 0.7 are $N_H=26$ and $N_H=10$, respectively. Therefore, both the increase in the number of HAPSs and the decrease in the 3 dB beamwidth of directional antennas can improve the UL coverage performance. In practical scenarios, the number of HAPSs and beamwidth of directional antennas should be designed by jointly considering DL performance, UL performance, system complexity, and operational expense. In practical deployments, key system parameters, including the TBS density, UE density, HAPS altitude, and carrier frequency, depend on the specific deployment scenario. Therefore, the performance should be evaluated under scenario-specific parameter settings.

In Fig. \ref{fig:UL_Covpro_hH}, we fix the number of HAPSs as $N_H=16$ and investigate the impact of HAPS altitude on the UL coverage performance of the UE located at $\textbf{o}_U$. When $\psi_b=10^\circ$ or $\psi_b=5^\circ$, increasing the HAPS altitude weakens the received signal power from the serving UE, so that the average UL coverage probability decreases. When $h_H=20\,{\rm km}$, the UL coverage probabilities are 0.20 and 0.54 for $\psi_b=10^\circ$ and $\psi_b=5^\circ$, respectively. However, when $\psi_b=2.5^\circ$, the average UL coverage probability remains higher than 0.8 for $20\,{\rm km}<h_H<50\,{\rm km}$, where the optimal altitude of HAPSs is $h_H=32\,{\rm km}$ and the corresponding average UL coverage probability is 0.86. We also fix $N_H=16$ and evaluate the impact of UE transmit power $p_U^{\rm Tx}$ on the UL performance in Fig. \ref{fig:UL_Covpro_pU}. To achieve a coverage probability above 0.7, the HAPSs equipped with directional antennas with $\psi_b=2.5^\circ$ require $p_U^{\rm Tx}=16\,{\rm dBm}$. However, when $\psi_b=5^\circ$ and $\psi_b=10^\circ$, the UL coverage probabilities are below 0.68 and 0.3, respectively, which are below the target requirement.

\begin{figure}[ht]
    \centering
    \includegraphics[width=0.88\linewidth]{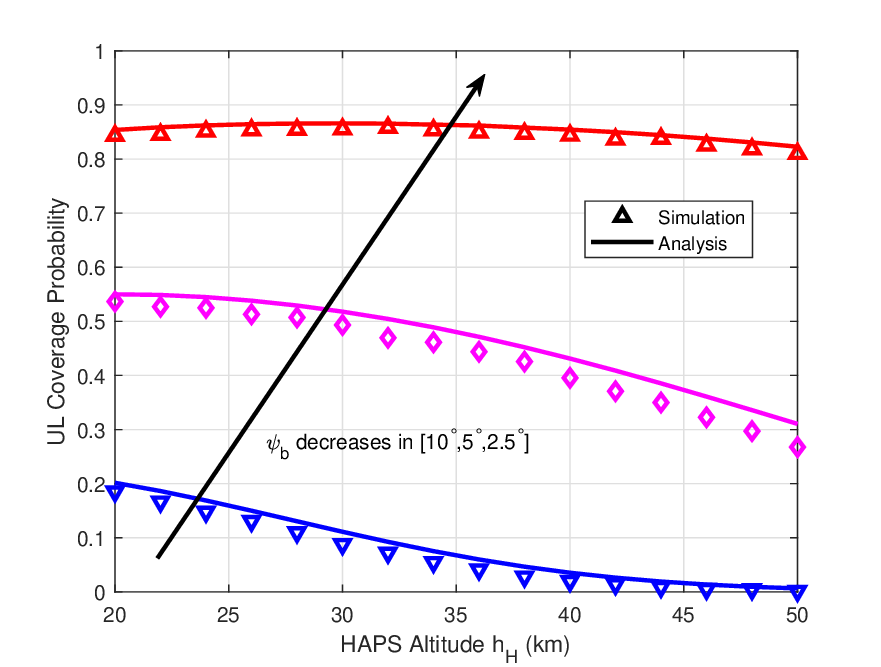}
    \caption{Average UL performance of the UE at the center of the hard-to-reach area versus the altitude of HAPSs $h_H$, when the HAPSs are equipped with directional antennas.}
    \label{fig:UL_Covpro_hH}
\end{figure}

\begin{figure}[ht]
    \centering
    \includegraphics[width=0.88\linewidth]{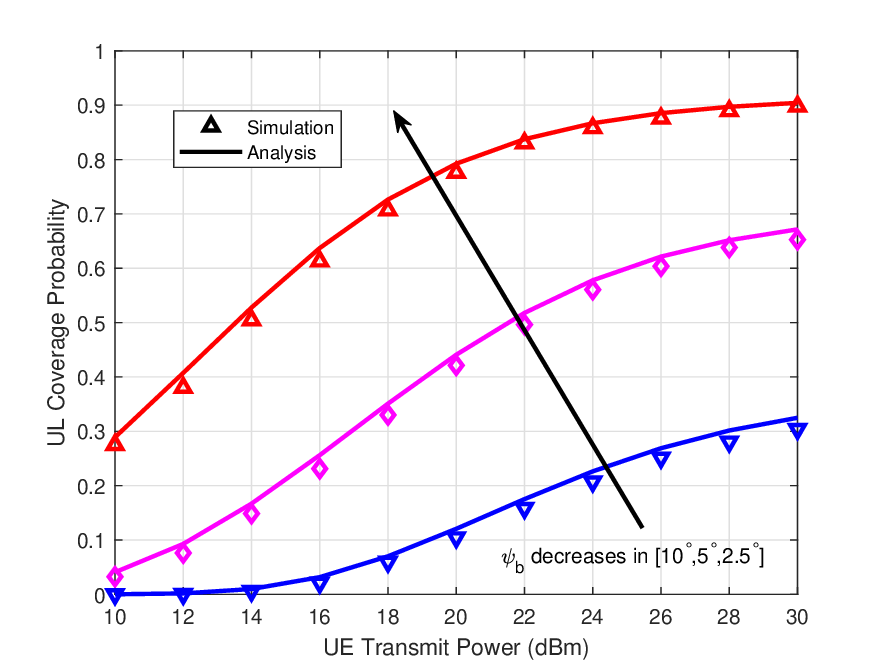}
    \caption{Average UL performance of the UE at the center of the hard-to-reach area versus transmit power of all UEs $p_U^{\rm Tx}$, when the HAPSs are equipped with directional antennas.}
    \label{fig:UL_Covpro_pU}
\end{figure}

In this paper, we model the hard-to-reach area as a circular region, which maintains both analytical tractability and representativeness. In real-world scenarios, hard-to-reach areas may be irregular due to terrain and environmental constraints. For areas that can be reasonably approximated by a circular region, the proposed model can be directly applied by selecting an appropriate effective radius. For more irregular areas, the proposed framework can be extended by modifying the distance distribution calculations according to the actual geometry, thereby enabling the evaluation of the coverage performance for users at different locations. We adopt the assumption that each user connects to its closest HAPS or closest TBS that provides the strongest average received power. This association rule provides a tractable baseline for evaluating the coverage performance. However, in realistic HAPS-terrestrial integrated networks, user association may also be affected by the backhaul constraints of each HAPS, load balancing, resource allocation among different HAPSs and TBSs. For mobile users and devices, the handover management among HAPSs and TBSs should also be addressed. Moreover, efficient scheduling for massive numbers of devices and users should be addressed in future work to reduce mutual interference among users and enhance spectral efficiency.

In this paper, we compare the coverage performance when HAPSs adopt either omnidirectional or directional antenna patterns, while TBSs are assumed to be equipped with omnidirectional antennas. This assumption provides a tractable baseline for evaluating the impact of HAPS directional beamforming on the coverage performance and required HAPS number. In practical heterogeneous systems, TBSs may also employ directional antenna patterns, which may affect HAPS/TBS association and terrestrial interference management.
In addition, HAPSs can be deployed not only in hard-to-reach areas but also in urban areas to enhance the quality of service for users. In such scenarios, the relative locations of users, HAPSs, and buildings may significantly affect the channel conditions. Therefore, altitude-dependent LoS probability models discussed in \cite{6863654} can be incorporated.

\section{Conclusion}
In this paper, we built a mathematical framework based on stochastic geometry to investigate the DL and UL coverage performance of a HAPS-based solution for hard-to-reach areas. We considered the coexistence of a HAPS network over the hard-to-reach area and a TBS network deployed only around the perimeter due to environmental and geographical constraints. We first analyzed the conditional DL performance and then derived the average DL coverage probability of UEs. Next, we studied the UL performance conditioned on the HAPS locations, and then derived the average UL coverage probability. We conducted Monte Carlo simulations to validate the analytical results. Our results showed that deploying a sufficient number of HAPSs can help fill the coverage hole in the hard-to-reach area and avoid the connection gaps caused by handover between HAPSs and TBSs. For UL services, we generated coverage maps based on the HAPS location information and further evaluated the average UL performance. The results also showed that narrower HAPS beams can significantly improve the coverage probability and reduce the number of HAPSs required to satisfy a target coverage probability. Therefore, the number of HAPSs and the 3 dB beamwidth should be jointly designed by considering the DL performance, UL performance, system complexity, and operational expense.

\appendices
% \clearpage
\section{Proof of Theorem \ref{theo:DLCP_dB}}\label{app:DLCP_dB}

In \eqref{SINRDLB}, we have formulated the DL SINR when the UE is associated with its closest TBS. We define $U_T=N_0+I_{T}^{\backslash 1}+I_H$ as a random variable representing the sum of interference and noise. $I_{T}^{\backslash 1}$ represents the interference from other TBSs except for the associated one, $I_H$ represents the interference from all HAPSs and $N_0$ represents the thermal noise power. Therefore, the DL SINR coverage probability when the UE is associated with its closest TBS can be calculated using:
\begin{equation}
\begin{array}{r@{}l}
    P&_{\rm DL}^T(\tau_{\rm DL}|d_T)=\mathbb{P}\{ {\rm SINR}_{\rm DL}^T>\tau_{\rm DL},Q=T|d_T\}\\
    &=\mathbb{P}\{Q=T|d_T\}\mathbb{P}\{ {\rm SINR}_{\rm DL}^T>\tau_{\rm DL}|Q=T,d_T\}\\
    &=\overline{F}_{D_H}(t_{T2H}(d_T))\mathbb{P}\{ {\rm SINR}_{\rm DL}^T>\tau_{\rm DL}|Q=T,d_T\}.
\end{array}
\end{equation}

Because $W_T^1$ follows the Gamma distribution with shape parameter $m_T$, we have:
\begin{equation}
\begin{array}{r@{}l}
    \mathbb{P}&\{ {\rm SINR}_{\rm DL}^T>\tau_{\rm DL}|Q=T,d_T\}\\
    &=\mathbb{E}_{U_T}\bigg[\mathbb{P}\bigg\{\frac{ p_TW_T^1 {d_T^1}^{-\alpha_T}}{U_T}>\tau_{\rm DL}\bigg\}\bigg]\\
    &=\mathbb{E}_{U_T}\bigg[\mathbb{P}\bigg\{W_T^1>\frac{\tau_{\rm DL} {d_T^1}^{\alpha_T}}{p_T } U_T\bigg\}\bigg]\\
    &\displaystyle=\mathbb{E}_{U_T}\bigg[\frac{\Gamma_u\bigg(m_T,m_T \frac{\tau_{\rm DL} {d_T^1}^{\alpha_T}}{p_T } U_T\bigg)}{\Gamma(m_T)}\bigg]\\
    &\displaystyle = \mathbb{E}_{U_T}\bigg[ \exp(-s_T U_T)\sum_{n=0}^{m_T-1}\frac{(s_TU_T)^n}{n!} \bigg]\\
    &\displaystyle = \sum_{n=0}^{m_T-1} \frac{s_T^n}{n!}\mathbb{E}_{U_T}\bigg[\exp(-s_T U_T)U_T^n\bigg]\\
    &\displaystyle =\sum_{n=0}^{m_T-1}\frac{s_T^n}{n!}(-1)^{n}\frac{\partial^n}{\partial s_T^n}\mathcal{L}_{U_T}(s_T|d_T),
\end{array}
\end{equation}
where we let $s_T=\frac{m_T \tau_{\rm DL} {d_T^1}^{\alpha_T}}{p_T}$. Since the channel gains between the UE and all APs are independent, and the point process of HAPSs and the point process of TBSs are also independent of each other, the Laplace transform of the total interference $U_T=N_0+I_T^{\backslash 1}+I_{H}$ can be calculated using:
\begin{equation}
    \begin{array}{r@{}l}
        &\mathcal{L}_{N_0+I_{T}^{\backslash 1}+I_H}(s_T|d_T)\\
        &=\mathcal{L}_{N_0}(s_T)\mathcal{L}_{I_{T}^{\backslash 1}}(s_T|d_T)\mathcal{L}_{I_H}(s_T|t_{T2H}(d_T)).
    \end{array}
    \end{equation}
The Laplace transform of thermal noise is $\mathbb{E}_{N_0}[\exp(-s_TN_0)]=\exp(-s_TN_0)$. Because the TBSs follow a PPP with density $\lambda_T$, the Laplace transform of the interference from the TBSs except for the associated one can be calculated using the probability generating functional (PGFL) of a Poisson point process \cite{haenggi2013stochastic}:
\begin{equation}
    \begin{array}{r@{}l}
    &\mathcal{L}_{I_{T}^{\backslash 1}}(s_T|d_T)\displaystyle\\
    &\overset{\triangle}{=}\mathbb{E}_{\{W_T^j\},\Phi_T\backslash \{\textbf{z}_T^1\}}\bigg[e^{-s_T\sum\limits_{\textbf{z}_T^{j}\in\Phi_T\backslash \{\textbf{z}_T^{1}\}}p_T W_T^j {d_T^j}^{-\alpha_T} }\bigg|d_T\bigg]\\
    &=\displaystyle\mathbb{E}_{\Phi_T\backslash \{\textbf{z}_T^1\}}\bigg[ \prod_{\textbf{z}_T^{j}\in\Phi_T\backslash \{\textbf{z}_T^1\}} \mathbb{E}_{W_T^j} \bigg[ e^{-s_T p_T W_T^j {d_T^j}^{-\alpha_T} }\bigg] \bigg|d_T\bigg]\\
    
    &=\displaystyle\mathbb{E}_{\Phi_T\backslash \{\textbf{z}_T^1\}}\displaystyle\bigg[\prod_{\textbf{z}_T^{j}\in\Phi_T\backslash \{\textbf{z}_T^1\}}\bigg(1+\frac{s_T  p_T {d_T^j}^{-\alpha_T}}{m_T}\bigg)^{-m_T}\bigg|d_T\bigg]\\
    &\displaystyle=\exp\bigg(-\int_{d_T}^{d_T^{\rm max}}\int_{-\pi}^{\pi} \lambda_T \\
    &\;\;\;\;\;\;\;\;\;\;\;\;\times \mathbbm{1}[r_{\rm in}\leq \sqrt{r_U^2+r^2+2r_U r \cos \theta} \leq r_{\rm out}] \\
    &\displaystyle\;\;\;\;\;\;\;\;\;\;\;\;\times \left. \bigg[1-\bigg(1+\frac{s_T p_T r^{-\alpha_T}}{m_T}\bigg)^{-m_T}\bigg]{\rm d}\theta r{\rm d}r\right)\\
    &\displaystyle=\exp\bigg(-\int_{d_T}^{d_T^{\rm max}}\lambda_T 2\theta_T(r) \\
    &\displaystyle\;\;\;\;\;\;\;\;\;\;\;\;\;\;\;\;\;\;\;\;\;\;\;\;\;\;\;\;\;\;\times\bigg[1-\bigg(1+\frac{s_T p_T r^{-\alpha_T}}{m_T}\bigg)^{-m_T}\bigg] r{\rm d}r\bigg).
\end{array}
\end{equation}

Also, because the locations of HAPSs follow a BPP, we can use the conclusion from \cite{7882710}. We denote $D$ as the distance from the UE to an arbitrary HAPS conditioned on being farther than $t_{T2H}(d_T)$, and its conditional PDF $f_{D}(d|t_{T2H}(d_T))$ is shown in \eqref{modifiedPDF}. Since the minimum distance to the HAPSs is $t_{T2H}(d_T)$, we can calculate the Laplace transform of the interference from HAPSs as:
\begin{equation}
\begin{array}{r@{}l}
    &\mathcal{L}_{I_H}(s_T|t_{T2H}(d_T))\\
    &\displaystyle\overset{\triangle}{=}\mathbb{E}_{\{W_H^k\},\Phi_H,\{\psi_{\rm DL}^k\}}\\
    &\;\;\;\;\;\;\;\;\;\;\;\;\;\bigg[e^{-s_T\sum\limits_{\textbf{z}_H^{k}\in\Phi_H}p_{H} W_H^k {d_H^k}^{-\alpha_H} 10^{\frac{G_{\rm ITU}^{\rm dB}(\psi_{\rm DL}^k)-G_{\rm ITU}^{\rm dB}(0)}{10}} }\bigg]\\
    &=\displaystyle\mathbb{E}_{\Phi_H}\bigg[ \prod_{\textbf{z}_H^{k}\in\Phi_H}\\
    & \;\;\;\;\;\;\;\;\;\displaystyle\mathbb{E}_{W_H^k,\psi_{\rm DL}^k}\bigg[ e^{-s_T p_{H} W_H^k {d_H^k}^{-\alpha_H} 10^{\frac{G_{\rm ITU}^{\rm dB}(\psi_{\rm DL}^k)-G_{\rm ITU}^{\rm dB}(0)}{10}} }\bigg] \bigg]\\
    &=\displaystyle\mathbb{E}_{\Phi_{H}}\bigg[\prod_{\textbf{z}_{H}^{k}\in\Phi_{H}}\int_{t_{T2H}(d_T)}^{d_H^{\rm max}}\int_0^{\pi}\\
    &\;\;\;\;\;\;\;\;\;\;\;\;\;\;\bigg(1+\frac{s_T  p_{H} {d_H^k}^{-\alpha_H} }{m_{H}}10^{\frac{G_{\rm ITU}^{\rm dB}(\psi)-G_{\rm ITU}^{\rm dB}(0)}{10}}\bigg)^{-m_{H}}\\
    &\;\;\;\;\;\;\;\;\;\;\;\;\;\;\;\;\;\;\;\;\;\;\;\;\;\;\;\;\;\;\;\;\;\;\displaystyle\times f_{\psi}(\psi)f_{D}(d|t_{T2H}(d_T)){\rm d}\psi{\rm d}d\bigg]\\
    &\displaystyle=\bigg[\int_{t_{T2H}(d_T)}^{d_H^{\rm max}}\int_0^{\pi}\bigg(1+\frac{s_T  p_{H} d^{-\alpha_{H}}}{m_{H}}10^{\frac{G_{\rm ITU}^{\rm dB}(\psi)-G_{\rm ITU}^{\rm dB}(0)}{10}}\bigg)^{-m_H}\\
    &\displaystyle\;\;\;\;\;\;\;\;\;\;\;\;\;\;\times\frac{1}{\pi}f_{D}(d|t_{T2H}(d_T)){\rm d}\psi{\rm d}d\bigg]^{N_H}.
\end{array}
\end{equation}

It is worth noting that if the downlink beam directions of interfering HAPSs were independently and isotropically distributed in 3-D space, the off-boresight angles $\psi_{\rm DL}^k$ would be independent and identically distributed (i.i.d.) variables with the PDF $f_{\psi}(\psi)=\frac{1}{2}\sin\psi, 0\leq\psi\leq\pi$. However, this isotropic model is not fully consistent with the considered HAPS transmission scenario, where HAPS beams are steered toward their scheduled ground UEs. Moreover, the distribution of scheduled users is affected by the HAPS locations and AP association, which would lead to multi-layer integration when deriving the exact off-boresight angle distribution.
To maintain analytical tractability and preserve the main coverage insights, we approximate the off-boresight angles $\psi_{\rm DL}^k$ as independent random variables with a common marginal distribution. In the absence of exact prior statistics for the beam directions, we adopt the uniform angular approximation  
$f_{\psi}(\psi)=\frac{1}{\pi}, 0\leq\psi\leq\pi$
as a tractable baseline model. 

As shown in \cite{6932503,8833522,10050345}, the tight upper bound on the upper incomplete Gamma function has been widely adopted to obtain tractable approximations of coverage probability under Nakagami-$m$ fading. Following these works, we use this approximation to avoid the high-order derivatives of the Laplace transform, and the downlink coverage probability can be approximated as:
\begin{equation}
\begin{array}{r@{}l}
    \mathbb{P}&\{{\rm SINR}_{\rm DL}^T>\tau_{\rm DL}|Q=T,d_T\}\\
    &\displaystyle=\mathbb{E}_{U_T}\bigg[\frac{\Gamma_u\bigg(m_T,m_T \frac{\tau_{\rm DL} {d_T^1}^{\alpha_T}}{p_T } U_T\bigg)}{\Gamma(m_T)}\bigg]\\
    &\displaystyle\approx\mathbb{E}_{U_T}\bigg[1-(1-e^{-\beta_2^T s_T U_T})^{m_T}\bigg]\\
    &\displaystyle =\mathbb{E}_{U_T}\bigg[\sum_{n=1}^{m_T}\binom{m_T}{n} (-1)^{n+1}\exp(-n\beta_2^T s_T U_T) \bigg]\\
    &\displaystyle =\sum_{n=1}^{m_T}\binom{m_T}{n} (-1)^{n+1}\mathbb{E}_{U_T}\bigg[\exp(-n\beta_2^T s_T U_T) \bigg]\\
    &\displaystyle =\sum_{n=1}^{m_T}\binom{m_T}{n} (-1)^{n+1}\mathcal{L}_{U_T}(n\beta_2^Ts_T|d_T),
\end{array}
\end{equation}
where $\beta_2^T=(m_T!)^{-\frac{1}{m_T}}$.

% \newpage
\section{Proof of Theorem \ref{theo:ULCP_oneshot}}\label{app:ULCP_oneshot}
Conditioned on the location of the closest HAPS, the UL coverage probability of the HAPS network at the UE is calculated using:
\begin{equation}
\begin{array}{r@{}l}
    P&_{\rm UL}^{H}(\tau_{\rm UL}|\textbf{z}_H^1)=\mathbb{P}\{ {\rm SINR}_{\rm UL}^{H}>\tau_{\rm UL},Q=H|\textbf{z}_H^1\}\\
    &=\mathbb{P}\{Q=H|\textbf{z}_H^1\}\mathbb{P}\{ {\rm SINR}_{\rm UL}^{H}>\tau_{\rm UL}|Q=H,\textbf{z}_H^1\}\\
    &=\overline{F}_{D_T}(t_{H2T}(\|\textbf{z}_H^1-\textbf{z}_U^0\|))\mathbb{P}\{ {\rm SINR}_{\rm UL}^{H}>\tau_{\rm UL}|Q=H,\textbf{z}_H^1\}.
\end{array}
\end{equation}
Let $U_U$ denote the sum of interference from grant-free interfering UEs $I_U$ and the thermal noise $N_0$, that is $U_U=I_U+N_0$. We denote $d_U^0=\|\textbf{z}_H^1-\textbf{z}_U^0\|$. We can calculate the conditional UL coverage probability as:
\begin{equation}
\begin{array}{r@{}l}

    \mathbb{P}&\{ {\rm SINR}_{\rm UL}^{H}>\tau_{\rm UL}|Q=H,\textbf{z}_H^1\}\\
    % &=\mathbb{E}_{U_U}\bigg[\mathbb{P}\bigg\{\frac{ p_U W_U^0 {d_U^0}^{-\alpha_H}}{U_U}>\tau_{\rm UL}\bigg\}\bigg]\\
    &=\mathbb{E}_{U_U}\bigg[\mathbb{P}\bigg\{W_U^0>\frac{\tau_{\rm UL} {d_U^0}^{\alpha_H}}{p_U } U_U\bigg\}\bigg]\\
    &=\mathbb{E}_{U_U}\bigg[\displaystyle\frac{\Gamma_u\bigg(m_{H},m_{H} \frac{\tau_{\rm UL} {d_U^0}^{\alpha_H}}{p_U } U_U\bigg)}{\Gamma(m_H)}\bigg]\\
    &\displaystyle = \mathbb{E}_{U_U}\bigg[ \exp(-s_U U_U)\sum_{n=0}^{m_H-1}\frac{(s_U U_U)^n}{n!} \bigg]\\
    &\displaystyle = \sum_{n=0}^{m_H-1} \frac{s_U^n}{n!}\mathbb{E}_{U_U}\bigg[\exp(-s_U U_U)U_U^n\bigg]\\
    &\displaystyle =\sum_{n=0}^{m_H-1}\frac{s_U^n}{n!}(-1)^{n}\frac{\partial^n}{\partial s_U^n}\mathcal{L}_{U_U}(s_U|\textbf{z}_H^1),
\end{array}
\end{equation}
where $s_U=\frac{m_H \tau_{\rm UL} {d_U^0}^{\alpha_H}}{p_U}$ and $\mathcal{L}_{U_U}(s_U|\textbf{z}_H^1)=\mathcal{L}_{I_U}(s_U|\textbf{z}_H^1)\mathcal{L}_{N_0}(s_U)$. The Laplace transform of thermal noise is $\mathcal{L}_{N_0}(s_U)=\exp(-s_U N_0)$, and the Laplace transform of the interference from grant-free interfering UEs is
\begin{equation}
\begin{array}{r@{}l}
    &\mathcal{L}_{I_U}(s_U|\textbf{z}_H^1)\\
    &\displaystyle\overset{\triangle}{=}\mathbb{E}_{\{W_U^i\},\Phi_U}\bigg[e^{-s_U\sum_{\textbf{z}_U^{i}\in\Phi_U}p_{U} G_{\rm UL}^{i} W_U^i {d_U^i}^{-\alpha_H} }\bigg|\textbf{z}_H^1\bigg]\\
    &=\displaystyle\mathbb{E}_{\Phi_U}\bigg[\prod_{\textbf{z}_U^{i}\in\Phi_U} \mathbb{E}_{W_U^i} \bigg[ e^{-s_U p_{U} G_{\rm UL}^{i} W_U^i {d_U^i}^{-\alpha_H} } \bigg] \bigg|\textbf{z}_H^1\bigg]\\
    &=\displaystyle\mathbb{E}_{\Phi_U}\bigg[\prod_{\textbf{z}_{U}^{i}\in\Phi_U} \bigg(1+\frac{s_U  p_{U} G_{\rm UL }^{i}{d_U^i}^{-\alpha_{H}}}{m_{H}}\bigg)^{-m_{H}} \bigg]\bigg|\textbf{z}_H^1\bigg]\\
    &\displaystyle = \exp\bigg(-\int_{\textbf{b}(\textbf{o}_U,r_{\rm in})}  \lambda_U \\
    &\;\;\times\bigg[1-\bigg(1+\frac{s_U  p_{U} {d_U^i}^{-\alpha_{H}}G_{\rm ITU}^{\rm linear}(r_U^i,\theta_U^i|\textbf{z}_H^1)}{m_{H}} \bigg)^{-m_H}\bigg]{\rm d}\textbf{z}_{U}^i\bigg)\\
    &\displaystyle = \exp\bigg(-\int_{0}^{r_{\rm in}} \int_{-\pi}^{\pi}  \lambda_U \bigg[ 1-\\
    &\displaystyle\;\;\bigg(1+\frac{s_U  p_{U} d(r,\theta|\textbf{z}_H^1)^{-\alpha_{H}}G_{\rm ITU}^{\rm linear}(r,\theta|\textbf{z}_H^1)}{m_{H}} \bigg)^{-m_H}\bigg]r{\rm d}\theta{\rm d}r\bigg),
\end{array}
\end{equation}
where
\begin{equation}
    d(r,\theta|\textbf{z}_H^1)=\sqrt{h_H^2+(r_H^1)^2+r^2-2r_H^1 r \cos (\theta_H^1-\theta)},
\end{equation}
and
\begin{equation}
    G_{\rm ITU}^{\rm linear}(r,\theta|\textbf{z}_H^1)=10^{\frac{G_{\rm ITU}^{\rm dB}(\psi_{\rm UL}(r,\theta|\textbf{z}_H^1))-G_{\rm ITU}^{\rm dB}(0)}{10}}.
\end{equation}

\ifCLASSOPTIONcaptionsoff
  \newpage
\fi

\bibliographystyle{IEEEtran}
\bibliography{ref}

\end{document}

%% file: notation.tex
% Bold lowercase: syntax \nb# where # is {a ... z, 0,1}
\def\nba{{\mathbf{a}}}
\def\nbb{{\mathbf{b}}}
\def\nbc{{\mathbf{c}}}
\def\nbd{{\mathbf{d}}}
\def\nbe{{\mathbf{e}}}
\def\nbf{{\mathbf{f}}}
\def\nbg{{\mathbf{g}}}
\def\nbh{{\mathbf{h}}}
\def\nbi{{\mathbf{i}}}
\def\nbj{{\mathbf{j}}}
\def\nbk{{\mathbf{k}}}
\def\nbl{{\mathbf{l}}}
\def\nbm{{\mathbf{m}}}
\def\nbn{{\mathbf{n}}}
\def\nbo{{\mathbf{o}}}
\def\nbp{{\mathbf{p}}}
\def\nbq{{\mathbf{q}}}
\def\nbr{{\mathbf{r}}}
\def\nbs{{\mathbf{s}}}
\def\nbt{{\mathbf{t}}}
\def\nbu{{\mathbf{u}}}
\def\nbv{{\mathbf{v}}}
\def\nbw{{\mathbf{w}}}
\def\nbx{{\mathbf{x}}}
\def\nby{{\mathbf{y}}}
\def\nbz{{\mathbf{z}}}
\def\nb0{{\mathbf{0}}}
\def\nb1{{\mathbf{1}}}

% Bold capital letters: syntax \nb# where # is {A ... Z}
\def\nbA{{\mathbf{A}}}
\def\nbB{{\mathbf{B}}}
\def\nbC{{\mathbf{C}}}
\def\nbD{{\mathbf{D}}}
\def\nbE{{\mathbf{E}}}
\def\nbF{{\mathbf{F}}}
\def\nbG{{\mathbf{G}}}
\def\nbH{{\mathbf{H}}}
\def\nbI{{\mathbf{I}}}
\def\nbJ{{\mathbf{J}}}
\def\nbK{{\mathbf{K}}}
\def\nbL{{\mathbf{L}}}
\def\nbM{{\mathbf{M}}}
\def\nbN{{\mathbf{N}}}
\def\nbO{{\mathbf{O}}}
\def\nbP{{\mathbf{P}}}
\def\nbQ{{\mathbf{Q}}}
\def\nbR{{\mathbf{R}}}
\def\nbS{{\mathbf{S}}}
\def\nbT{{\mathbf{T}}}
\def\nbU{{\mathbf{U}}}
\def\nbV{{\mathbf{V}}}
\def\nbW{{\mathbf{W}}}
\def\nbX{{\mathbf{X}}}
\def\nbY{{\mathbf{Y}}}
\def\nbZ{{\mathbf{Z}}}

% \mathcal: syntax \ncal# where # is {A ... Z}
\def\ncalA{{\mathcal{A}}}
\def\ncalB{{\mathcal{B}}}
\def\ncalC{{\mathcal{C}}}
\def\ncalD{{\mathcal{D}}}
\def\ncalE{{\mathcal{E}}}
\def\ncalF{{\mathcal{F}}}
\def\ncalG{{\mathcal{G}}}
\def\ncalH{{\mathcal{H}}}
\def\ncalI{{\mathcal{I}}}
\def\ncalJ{{\mathcal{J}}}
\def\ncalK{{\mathcal{K}}}
\def\ncalL{{\mathcal{L}}}
\def\ncalM{{\mathcal{M}}}
\def\ncalN{{\mathcal{N}}}
\def\ncalO{{\mathcal{O}}}
\def\ncalP{{\mathcal{P}}}
\def\ncalQ{{\mathcal{Q}}}
\def\ncalR{{\mathcal{R}}}
\def\ncalS{{\mathcal{S}}}
\def\ncalT{{\mathcal{T}}}
\def\ncalU{{\mathcal{U}}}
\def\ncalV{{\mathcal{V}}}
\def\ncalW{{\mathcal{W}}}
\def\ncalX{{\mathcal{X}}}
\def\ncalY{{\mathcal{Y}}}
\def\ncalZ{{\mathcal{Z}}}

% \mathbb: syntax \nbb# where # is {A ... Z}
\def\nbbA{{\mathbb{A}}}
\def\nbbB{{\mathbb{B}}}
\def\nbbC{{\mathbb{C}}}
\def\nbbD{{\mathbb{D}}}
\def\nbbE{{\mathbb{E}}}
\def\nbbF{{\mathbb{F}}}
\def\nbbG{{\mathbb{G}}}
\def\nbbH{{\mathbb{H}}}
\def\nbbI{{\mathbb{I}}}
\def\nbbJ{{\mathbb{J}}}
\def\nbbK{{\mathbb{K}}}
\def\nbbL{{\mathbb{L}}}
\def\nbbM{{\mathbb{M}}}
\def\nbbN{{\mathbb{N}}}
\def\nbbO{{\mathbb{O}}}
\def\nbbP{{\mathbb{P}}}
\def\nbbQ{{\mathbb{Q}}}
\def\nbbR{{\mathbb{R}}}
\def\nbbS{{\mathbb{S}}}
\def\nbbT{{\mathbb{T}}}
\def\nbbU{{\mathbb{U}}}
\def\nbbV{{\mathbb{V}}}
\def\nbbW{{\mathbb{W}}}
\def\nbbX{{\mathbb{X}}}
\def\nbbY{{\mathbb{Y}}}
\def\nbbZ{{\mathbb{Z}}}

% \mathfrak:
\def\nfrakR{{\mathfrak{R}}}

% Roman: {\rm } syntax \nrm# where # is {a ... z}
\def\nrma{{\rm a}}
\def\nrmb{{\rm b}}
\def\nrmc{{\rm c}}
\def\nrmd{{\rm d}}
\def\nrme{{\rm e}}
\def\nrmf{{\rm f}}
\def\nrmg{{\rm g}}
\def\nrmh{{\rm h}}
\def\nrmi{{\rm i}}
\def\nrmj{{\rm j}}
\def\nrmk{{\rm k}}
\def\nrml{{\rm l}}
\def\nrmm{{\rm m}}
\def\nrmn{{\rm n}}
\def\nrmo{{\rm o}}
\def\nrmp{{\rm p}}
\def\nrmq{{\rm q}}
\def\nrmr{{\rm r}}
\def\nrms{{\rm s}}
\def\nrmt{{\rm t}}
\def\nrmu{{\rm u}}
\def\nrmv{{\rm v}}
\def\nrmw{{\rm w}}
\def\nrmx{{\rm x}}
\def\nrmy{{\rm y}}
\def\nrmz{{\rm z}}

% Special symbols
\def\nbydef{:=}
\def\nborel{\ncalB(\nbbR)}
\def\nboreld{\ncalB(\nbbR^d)}
\def\sinc{{\rm sinc}}

% Theorems etc.
\newtheorem{lemma}{Lemma}
\newtheorem{thm}{Theorem}
\newtheorem{definition}{Definition}
\newtheorem{ndef}{Definition}
\newtheorem{nrem}{Remark}
\newtheorem{theorem}{Theorem}
\newtheorem{prop}{Proposition}
\newtheorem{cor}{Corollary}
\newtheorem{example}{Example}
\newtheorem{remark}{Remark}
\newtheorem{assumption}{Assumption}
	
%%%%%%%% Backwards compatibility

\newcommand{\ceil}[1]{\lceil #1\rceil}
\def\argmin{\operatorname{arg~min}}
\def\argmax{\operatorname{arg~max}}
\def\figref#1{Fig.\,\ref{#1}}%
\def\E{\mathbb{E}}
\def\EE{\mathbb{E}^{!o}}
\def\P{\mathbb{P}}
\def\pc{\mathtt{P_c}}
\def\rc{\mathtt{R_c}}   % rate coverage
\def\p{p}

\def\V{\operatorname{Var}}
\def\erfc{\operatorname{erfc}}
\def\erf{\operatorname{erf}}
\def\opt{\mathrm{opt}}
\def\R{\mathbb{R}}
\def\Z{\mathbb{Z}}

\def\LL{\mathcal{L}^{!o}}
\def\var{\operatorname{var}}
\def\supp{\operatorname{supp}}

\def\N{\sigma^2}
\def\T{\beta}							% Threshold = \beta_i
\def\sinr{\mathtt{SINR}}			% Signal to interference plus noise ratio
\def\snr{\mathtt{SNR}}
\def\sir{\mathtt{SIR}}
\def\ase{\mathtt{ASE}}
\def\se{\mathtt{SE}}

\def\calN{\mathcal{N}}
\def\FE{\mathcal{F}}
\def\calA{\mathcal{A}}
\def\calK{\mathcal{K}}
\def\calT{\mathcal{T}}
\def\calB{\mathcal{B}}
\def\calE{\mathcal{E}}
\def\calP{\mathcal{P}}
\def\calL{\mathcal{L}}
%\DeclareMathOperator{\Tr}{Tr}
%\DeclareMathOperator{\rank}{rank}
%\DeclareMathOperator{\Pois}{Pois}

%\DeclareMathOperator{\TC}{\mathtt{TC}}
%\DeclareMathOperator{\TCL}{\mathtt{TC_l}}
%\DeclareMathOperator{\TCU}{\mathtt{TC_u}}

% Fading
\def\l{\ell}
\newcommand{\fad}[2]{\ensuremath{\mathtt{h}_{#1}[#2]}}
\newcommand{\h}[1]{\ensuremath{\mathtt{h}_{#1}}}

\newcommand{\err}[1]{\ensuremath{\operatorname{Err}(\eta,#1)}}
\newcommand{\FD}[1]{\ensuremath{|\mathcal{F}_{#1}|}}

%% Symbols changed
% \def\i{\mathbf{1}}					% changed to \nb1
% \def\d{\mathrm{d}}					% changed to \nrmd
% \def\L{\mathcal{L}}					% changed to \ncalL
% \begin{definition}					% changed to \begin{ndef}

% \l also gives problems. Use \ell after defining it if needed.

%% D2D def
\def\Bx{{\mathcal{B}}^x}
\def\Bxx{{\mathcal{B}}^{x_0}}
\def\jx{y}
\def\m{(\bar{n}-1)}
\def\mm{\bar{n}-1}
\def\Nx{{\mathcal{N}}^x}
\def\Nxo{{\mathcal{N}}^{x_0}}
\def\wj{w_{jx_0}}
\def\uij{u_{jx}}
 \def\yj{y}
 \def\yjx{y}
 \def\zjx{z_x}
 \def \tx {y_0}
 \def \htx {h_0}

\def\rx{z_{1}}
\def\ry{z_{2}}

\def\Rx{Z_{1}}
\def\Ry{Z_{2}}

%% fading
\def \hyxx {h_{y_{x_0}}}
\def \hyx {h_{y_x}}

\def\nbb1{\mathbbm{1}}
\def\yi{\textbf{y}_i}
\def\yy{\textbf{y}_1}
\def\xx{\textbf{x}_0}
\def\wj{\textbf{w}_j}
\def\od{\textbf{o}_d}
\def\oc{\textbf{o}_c}
\def\ie{{\em i.e. }}
\def\eg{{\em e.g. }}
\def\iid{{\em i.i.d. }}
\def\gi{G_{\yi}}
\def\g1{G_{\yy}}
\def\gj{G_{\wj}}
\def\hi{H_{\yi}}
\def\h1{H_{\yy}}
\def\hj{H_{\wj}}
\def\avg{\rm avg}

\def\rmnuma{\rm\uppercase\expandafter{\romannumeral1}}
\def\rmnumb{\rm\uppercase\expandafter{\romannumeral2}}
\def\rmnumc{\rm\uppercase\expandafter{\romannumeral3}}
\def\rmnumd{\rm\uppercase\expandafter{\romannumeral4}}
\def\rmnume{\rm\uppercase\expandafter{\romannumeral5}}
\def\rmnumf{\rm\uppercase\expandafter{\romannumeral6}}